\documentclass[letterpaper, 11pt, onehalfspacing, leqno, renum]{article}

\usepackage{amssymb}
\usepackage{amsmath}
\usepackage{amsfonts}
\usepackage{amsthm}
\usepackage{graphicx}
\usepackage{enumerate}
\usepackage[margin=1in]{geometry}
\usepackage{setspace}
\usepackage[applemac]{inputenc}
\usepackage{dsfont}
\usepackage{natbib}
\usepackage{mathrsfs}
\usepackage{color}
\usepackage{comment}
\usepackage{latexsym}
\usepackage{cancel}
\usepackage{hyperref}
\usepackage{nameref}
\hypersetup{breaklinks=true,
  colorlinks   = true, %Colours links instead of ugly boxes
  urlcolor     = blue, %Colour for external hyperlinks
  linkcolor    = {blue!60!black}, %Colour of internal links
  citecolor   = {red!40!black} %Colour of citations
}

\usepackage{geometry} % to change the page dimensions
\usepackage{xcolor} % support the \includegraphics command and options
\usepackage{tikz}
\usetikzlibrary{trees}
\usetikzlibrary{patterns}
\usepackage{subfig} % make it possible to include more than one captioned figure/table in a single float
\usepackage{thmtools, euscript} %basic math symbols, scripts, theorem environments
\usepackage{mathrsfs}
\usepackage{mhsetup, mathtools} %Must be used together; 

\newcommand{\R}{\mathbb{R}}

\newcommand{\ra}{\rightarrow}

\newtheorem{thm}{Theorem}
\newtheorem{prps}{Proposition}
\newtheorem{lem}{Lemma}

\theoremstyle{definition}
\newtheorem{axm}{\sc Axiom}
\newtheorem{ass}{Assumption}
\newtheorem{defn}{Definition}
\newtheorem{cor}{Corollary}

\newtheorem{obs}{Observation}

\begin{document}
\title{Behavioral Foundations of Nested Stochastic Choice and Nested Logit \footnote{Kovach: Virginia Tech (mkovach@vt.edu); Tserenjigmid: UC Santa Cruz (gtserenj@ucsc.edu ). We thank Jos\'e Apesteguia, Miguel A. Ballester, Khai Chiong, Ian Crawford, Federico Echenique, Mira Frick, Sean Horan, Ryota Iijima, Shaowei Ke, Jay Lu, Fabio Maccheroni, A.A.J. Marley, Paulo Natenzon, Pietro Ortoleva, Collin Raymond, Matthew Shum, Tomasz Strzalecki, the audiences of the BGSE Summer Forum 2018 and SAET 2019, and the seminar participants at UC Santa Cruz, University of Sussex, University of Toronto, KAIST, and Hitotsubashi University. We are also very grateful to the editor, Emir Kamenica, and four anonymous referees for their excellent suggestions and feedback.}}

\author{Matthew Kovach \and Gerelt Tserenjigmid}

\maketitle

\begin{abstract}
We provide the first behavioral characterization of \textbf{nested logit}, a foundational and widely applied discrete choice model, through the introduction of a non-parametric version of nested logit that we call \textbf{Nested Stochastic Choice} (NSC). NSC is characterized by a single axiom that weakens \textbf{Independence of Irrelevant Alternatives} based on revealed similarity to allow for the \textbf{similarity effect}. Nested logit is characterized by an additional menu-independence axiom. Our axiomatic characterization leads to a practical, data-driven algorithm that identifies the true nest structure from choice data. We also discuss limitations of generalizing nested logit by studying the testable implications of cross-nested logit. 
\end{abstract}

\bigskip


\section{Introduction}

\textbf{Nested logit} (\cite{Ben-Akiva1973nestedlogit}, \cite{mcfadden1978modeling}) is the most widely applied generalization of multinomial logit (or Luce's (1959) model) due to its ability to capture various substitution patterns.\footnote{Nested logit has been used to study transportation demand (\cite{anderson1992multiproduct}, \cite{forinash1993application}), airline competition (\cite*{lurkin2018}), automobile demand  (\cite{brownstone1989efficient}, \cite{goldberg1995product}), telephone use (\cite*{train1987demand}, \cite{lee1999calling}), and much more. See Chapter 4 of \cite{train2009discrete} for an excellent discussion.}  In nested logit, each alternative belongs to a nest (or subset) of ``similar" alternatives, and choice may be decomposed into two Luce procedures: the probability that $a$ is chosen from menu $A$ is the probability that $a$'s nest is chosen from among nests available in $A$ multiplied by the conditional probability that $a$ is chosen from that nest.  Despite its immense importance, nested logit has escaped behavioral characterization. In this paper, we provide the first behavioral characterization of nested logit through the introduction of a fully non-parametric version that we term \textbf{Nested Stochastic Choice} (NSC). This axiomatic characterization sheds light on the implicit assumptions behind nested logit and related models and leads to a tractable method to identify (unobserved) nests from data. 

Nested logit was developed to address the limitations of multinomial logit when dealing with ``similar alternatives.'' In the Luce model, choice probabilities are proportional to a utility index and hence satisfy \textbf{Independence of Irrelevant Alternatives} (IIA): probability ratios are menu-independent. The \textbf{similarity effect} (\cite{debreu1960review}, \cite{tversky1972eba}) is a violation of IIA in which the introduction of an alternative to a menu has a much larger effect on the choice probabilities of alternatives of a ``similar type'' than on those of a ``different type.'' Nested logit allows the similarity effect by assuming nested (similar) alternatives are more substitutable (e.g., because they receive a correlated utility shock). Our main result shows that nesting of similar alternatives, a key behavioral feature of nested logit (and also NSC), is captured by a weakening of IIA that allows for the similarity effect. This finding reveals a deep connection between nested logit and the similarity effect. %In fact, we show that a key behavioral feature of nested logit (or NSC) is characterized by a single weakening of IIA that allows for the similarity effect.

To illustrate the similarity effect, consider the red bus/blue bus example from \cite{debreu1960review}.  A commuter making a choice between a red bus and a train may choose either with probability $0.5$. If the option to take a blue bus is introduced, it is plausible that it will have no effect on the commuter's likelihood of taking the train; the blue bus only affects the probability of selecting the red bus (e.g., by reducing it to $0.25$).  The intuition behind this is that the buses are similar to each other in a way that neither one is to the train. Nested logit handles this by placing the two buses into a bus nest.  While this example provides an extreme case of the similarity effect (the buses are perfect substitutes, or ``duplicates''), the principle that ``similar alternatives affect each other'' readily extends to many situations of interest to economists, firms, and policy-makers, such as a consumer's choice of vehicle or apartment.\footnote{In these cases, the correct nest specification is not easily observed by the analyst. For instance, apartments in a city might be nested based on subjective neighborhoods, which may depend on a variety of factors. In turn, some of these factors may be observable while others may be subjective or difficult to observe.}

Our analysis of nested logit relies on the introduction of NSC, a non-parametric version of nested logit. Formally, a stochastic choice function $p$ is an NSC if there exist nests $X_1, \ldots, X_K$ that partition the set of all alternatives $X$ and functions $v$ and $u$ such that, for each choice set $A$, the probability of choosing $a \in A \cap X_i$ is given by\begin{equation}p(a,A)=\frac{v(A\cap X_i)}{\sum_{j=1}^K v(A\cap X_j)}\,\frac{u(a)}{\sum_{b \in A\cap X_i}u(b)}.\end{equation} 
The NSC is defined by two Luce rules where the attractiveness of the nest $A\cap X_i$ is measured by $v(A\cap X_i)$, and the attractiveness of the alternative $a$ is measured by $u(a)$ (or simply Luce's utility of $a$). In terms of the red bus/blue bus example, $v$ governs the choice between general modes, ``a bus'' or ``a train,'' while $u$ governs the choice between specific alternatives, the red or blue bus. 

Notice that nested logit is the special case of NSC in which, for each $i\le K$,
\begin{equation}\label{NLeq}v(A\cap X_i)=\left(\sum_{a \in A\cap X_i}u(a)\right)^{\eta_i}\text{ for some } \eta_i>0.\end{equation}
Hence, NSC is nested logit without any assumptions on the relationship between $v$ and $u$. It is commonly assumed in the applied literature that $\eta_i\le 1$, as this ensures that nested logit is a Generalized Extreme Value (GEV) model \citep{mcfadden1978modeling}, and therefore this restriction is sufficient for consistency with the random utility model (RUM). However, this restriction ($\eta_i\le 1$) is not necessary for \autoref{NLeq} to be a RUM and nested logit has been estimated without this restriction.\footnote{This parameter restriction is sometimes referred to as the Daly-Zachary-McFadden condition (\cite{Daly1978}, \cite{mcfadden1978modeling}), as they showed that this is sufficient for consistency with RUM for arbitrary values of the other variables (e.g., utilities). However, this restriction is not always imposed. For instance, \cite*{train1987demand} provide estimates of a model for which $\eta_i > 1$, remarking that it represents greater substitutability between nests than within nests (see also \cite*{train1989}, \cite{lee1999calling}, \cite{Foubert2007}). Further, nested logit with $\eta_i>1$ may still be consistent with RUM (see \cite{borsch1990} and \cite{herriges1996}).} Therefore, for simplicity of exposition we will refer to \autoref{NLeq} as a nested logit for any parameter value. To provide further clarity, we may sometimes refer to nested logit satisfying the restriction ($\eta_i\le 1$) as the \emph{random utility (RU) nested logit}. We provide behavioral foundations for NSC and nested logit, along with a characterization of random utility nested logit.

We utilize a revealed preference approach to identify the subjective/endogenous nest structure of the NSC.  This is achieved by introducing a notion of revealed similarity. In nested logit, if two alternatives are in the same nest, then their probability ratio will always be independent of other alternatives. Consistent with nested logit, we use this insight to define a notion of similarity: alternatives $a$ and $b$ are \textbf{revealed categorically similar}, denoted $a \sim_p b$, if IIA holds between $a$ and $b$ at any menu. Otherwise, they are \emph{revealed categorically dissimilar}.  Therefore, the core notion of similarity underpinning nested logit is binary: alternatives are similar or not. 

Equipped with this notion of reveled similarity, we can weaken IIA to allow for the similarity effect. To do so, we decompose IIA into two axioms. The first axiom, \textbf{\nameref{ISA}} (ISA), imposes an IIA condition between $a$ and $b$ in the presence of a third alternative $x$, when $x$ is symmetrically related to $a$ and $b$ in terms of revealed similarity (i.e., both are revealed categorically similar or dissimilar to $x$). In terms of the red bus/blue bus example, IIA should hold between the buses in the presence of the train. 
The second axiom, \textbf{\nameref{IAA}} (IAA), ``completes'' ISA by imposing an IIA condition when $x$ is asymmetrically similar to $a$ and $b$ (i.e., $x$ is similar to one and not the other). For example, IAA implies that the introduction of the blue bus impacts the red bus and the train equally, which directly rules out the similarity effect. 

Our main result is that ISA characterizes NSC (\autoref{NSCthm}). Further, since ISA is the minimal departure from IIA that allows for the similarity effect (i.e., violations of IAA), our analysis reveals that NSC is the model obtained when IIA is relaxed to allow for the similarity effect.

To see how our axiomatization provides a clearer picture of nested logit, recall its standard textbook description. For instance, Chapter 4 of \cite{train2009discrete} states that for nested logit ``IIA holds over alternatives in each nest and independence of irrelevant nests\footnote{If $a$ and $b$ are from distinct nests, then the addition of an alternative $c$ from a third nest will not affect the relative probabilities of $a$ and $b$.} (IIN) holds over alternatives in different nests.''  NSC also satisfies these properties. Indeed, ISA ensures the existence of endogenous nests and imposes exactly these properties on them. But this finding shows that ``IIA within a nest $+$ IIN'' are not sufficient for a nested logit representation, as ISA is equivalent to NSC. Since there are missing behavioral assumptions behind nested logit, the textbook description is incomplete.

To provide a complete picture of nested logit, we establish two characterizations of nested logit as well as a characterization of random utility nested logit. Our first characterization shows that an NSC is a nested logit if and only if it satisfies \textbf{\nameref{LRI}}. This axiom requires that the natural logarithm of certain probability ratios featuring collections of similar alternatives is menu-independent. The explicit use of a functional form in the axiom allows us to establish necessary and sufficient conditions for the functional form assumed in nested logit with finite data. 

Our second characterization is based upon a novel monotonicity condition, \textbf{\nameref{RLI}}, that is necessary for nested logit and becomes sufficient under a mild richness assumption.  To understand this axiom, note that nested logit (\autoref{NLeq}) requires that the attractiveness of a nest is increasing in the sum of utilities. \textbf{\nameref{RLI}} implies that this feature must hold in a relative sense; the attractiveness of a  nest \emph{relative to another nest} is increasing in the sum of utilities, holding the alternatives in the other nest fixed. Finally, random utility nested logit is characterized by one additional axiom, \textbf{Regularity}, a well-known monotonicity property that all random utility models must satisfy (under the same richness assumption).  We summarize all of our characterization results in \autoref{charsum}. 

\begin{table}[h]
\centering

\begin{tabular}{| l l l |}
\hline
 NSC  & $\Leftrightarrow$  ISA & \\
  \hline
 Luce  &  $\Leftrightarrow$  ISA + IAA  & $\Leftrightarrow$ IIA \\
 \hline
 Nested Logit & $\Leftrightarrow $ ISA + LRI & $\Leftrightarrow $ ISA + RLI (Richness) \\
  \hline
RU Nested Logit & $\Leftrightarrow $ & \,\,\,\,\,\,\,\,\, Nested Logit + Regularity (Richness) \\
\hline
\end{tabular}

\caption{Summary Of Characterization Results.}\label{charsum}
\end{table}

\nameref{ISA} implies that the revealed similarity relation $\sim_p$ is transitive, which ensures that the nests form a partition. Hence, our axiomatic characterizations show that the notion of (categorical) similarity in nested logit, as well as in NSC, is quite structured. In some applications, an analyst may want to allow for more flexible forms of substitutability. For this reason, \textbf{cross-nested logit}, a generalization of nested logit, has been proposed and widely applied in empirical work (see \cite{vovsha1997application}, \cite{ben1999discrete}). The main difference between nested logit and cross-nested logit is that alternatives may belong to several nests in cross-nested logit. Since we are focused on the problem of recovering the nesting structure, we consider a generalization of cross-nested logit that relaxes the typical parameter restriction, and refer to this as the \emph{unrestricted cross-nested logit}.  We show that the unrestricted cross-nested logit does not have testable implications. Therefore, our results reveal the trade-off between nested logit and cross-nested models: relaxing the partition structure of nested logit results in an overly permissive model. In other words, the behavioral content of cross-nested logit is essentially driven by the analyst's assumption of the nest structure and parameter restrictions.

In practice, $p$ is estimated from observed choices and ``IIA like" conditions never hold exactly. However, we show that the true nest structure can still be identified, for any NSC, by solving a minimization problem. In particular, our axiomatic characterization allows us to derive a ``distance" function $D$ that measures, for a given nest structure, the degree of violations of IIA within and across nests for a given set of observations. When the data are close to the true (or theoretical) $p$, the true nest structure will be the unique minimizer of $D$. In applied settings where the researcher has several nest structures in mind, $D$ may also be useful as a selection criteria.  

Because the number of possible nests grows exponentially as the number of alternatives increases, the full minimization problem may become intractable quickly. However, this issue can be managed due to insights from our similarity relation; one only needs to check nest structures that are consistent with an empirical approximation of $\sim_p.$ In fact, we show that there are at most $|X|$ potential nests that we need to check, where $|X|$ is the number of alternatives. We illustrate our theoretical finding and our data-driven algorithm to reduce the number number of candidate nests with a simulation exercise.

The rest of the paper is organized as follows. In \autoref{basics}, we discuss setup and notation as well as define NSC and nested logit. In \autoref{nestedchoice}, we define revealed similarity and the similarity effect (\ref{similarity}) before characterizing NSC (\ref{CharNSC}) and nested logit (\ref{nested}). We discuss ways of extending our notion of similarity and the testable implications of unrestricted cross-nested logit in \autoref{Disc}. The identification of nest structure from choice data is presented in \autoref{IDnest}. We conclude with a discussion of related literature in \autoref{rellit}. We also discuss the relationship between the similarity effect and regularity in Appendix \ref{monNSC}.

\section{Nested Stochastic Choice and Nested Logit}\label{basics}

All of our models are developed in the standard stochastic choice setup. Accordingly, let $X$ be a finite set of alternatives and $\mathscr{A}$ be the collection of all nonempty subsets of $X$ (menus). Let $\mathds{R}_{+}$ ($\mathds{R}_{++}$) denote the non-negative (positive) real numbers. 

\begin{defn} A function $p:X\times\mathscr{A}\to[0, 1]$ is a \textbf{stochastic choice function} if for any $A\in\mathscr{A}$, $\sum_{a\in A}p(a, A)=1$ and $p(x, A)=0$ when $x\not\in A$. In some instances, we may write $p(B,A)=\sum_{b \in B}p(b,A)$ for $B \in\mathscr{A}$. 
\end{defn}

Throughout this paper, we assume that $p$ is \emph{positive}; i.e., $p(a, A)>0$ for all $A\in\mathscr{A}$ and $a\in A$. For notational simplicity, we write $A\cup x$ instead of $A\cup \{x\}$. 

The Luce model is the most widely-known and influential stochastic choice model. In this model,  choice probabilities are proportional to a utility index: $p(a, A)=u(a)/\sum_{b\in A} u(b)$. In his seminal paper, \cite{luce1959individual} proves that a stochastic choice function can be represented by the Luce model if and only if it satisfies IIA for every pair of alternatives.

\begin{defn}[IIA] A stochastic  choice function \textbf{$p$ satisfies IIA at $a, b\in X$} if, for any $A\in\mathscr{A}$ with $a, b\in A$,
\[\frac{p(a,A)}{p(b,A)}=\frac{p(a,\{a,b\})}{p(b,\{a,b\})}.\]
Further, we say that \textbf{$p$ satisfies IIA} if $p$ satisfies IIA at any $a, b\in X$. 
\end{defn}

It is well known that IIA may fail when similar alternatives are added to the menu, as was illustrated by Debreu's (\citeyear{debreu1960review}) famous ``red bus/blue bus'' example. 
The nested logit is the most commonly applied generalization of Luce's model and was developed to accommodate violations of IIA like the similarity effect. We now formally define nested logit and the NSC, which is a novel, non-parametric version of nested logit.

\begin{defn}\label{NSC} A stochastic choice function $p$ is a \textbf{Nested Stochastic Choice} (NSC) if there exist a partition $X_1, \ldots, X_K$ of $X$, a utility function $u:X\to \mathds{R}_{++}$, and a nest utility function $v:\bigcup^K_{i=1} 2^{X_i} \to \mathds{R}_{+}$ with $v(\emptyset)=0$ such that for any $A\in\mathscr{A}$ and $a\in A\cap X_i$,
\begin{equation}
p(a, A)=\frac{v(A\cap X_i)}{\sum_{j=1}^{K}v(A\cap X_j)}\frac{u(a)}{\sum_{b\in A\cap X_i} u(b)}.
\end{equation}
Moreover, $p$ is a \textbf{nested logit} if there exist real numbers $\eta_1, \ldots, \eta_n>0$ such that for any $A\in\mathscr{A}$ and $i\le K$,\begin{equation}
v(A\cap X_i)=\Big(\sum_{x\in A\cap X_i}u(x)\Big)^{\eta_i}.\end{equation}
Finally, we say that $p$ is a \textbf{random utility (RU) nested logit} when $\eta_i\le 1$ for each $i\le K$. 
\end{defn} 

The NSC is defined by two Luce procedures, where $v$ governs the choice over nests (e.g., transportation modes or neighborhoods) and $u$ governs the choice over the particular alternatives in the selected nest (e.g., the red/blue bus or a specific apartment).  Note that the nest value function $v$ is not necessarily related to alternative utilities $u$, which enables the NSC to capture rich behavior (see \ref{bnl}). Despite this generality, the NSC may be falsified with relatively few observations. This is because behavior is disciplined by $u$ and the partition structure of the nests, both of which are menu-independent.\footnote{It is straightforward to derive from the representation that for any $A\subseteq X$ with $|A|=3$, there is a distinct pair $a, b\in A$ such that $\frac{p(a, \{a, b\})}{p(b, \{a, b\})}=\frac{p(a, A)}{p(b, A)}$. This is because for any three alternatives, either (i) at least two belong to the same nest or (ii) all three belong to distinct nests. Hence there must exist some pair for which IIA holds and therefore NSC may be rejected with only three alternatives.}

The nested logit imposes a specific parametric relationship between $v$ and $u$. Notice that the NSC, and consequently nested logit, reduces to the Luce model when there is a single nest.  Additionally, it is simple to see from \autoref{NSC} that any NSC satisfies ``IIA within a nest $+$ IIN." These properties are often taken as the hallmark of nested logit, yet they apply to all NSC (with endogenous nests). Since NSC permits behavior that nested logit does not (three examples are discussed in section \ref{bnl}), this means that there are additional behavioral assumptions underpinning nested logit. We elucidate these assumptions in section \ref{nested}.

In nested logit, $1-\eta_i$ is usually considered a measure of correlation or substitutability between alternatives in nest $i$.  When $\eta_i < 1$ alternatives within the same nest are substitutes. Further, it is well known that when $\eta_i < 1$, the nested logit is always a RUM for any profile of utilities. 

When $\eta_i >1$, choice frequencies may (but do not always) violate \emph{regularity}, a necessary property of every random utility model (RUM) which states that the probability of choosing some alternative must never increase as the menu expands.\footnote{There is some experimental evidence that violations of regularity occur when similar alternatives are introduced, in-line with convex aggregation. This has been observed in humans \citep*{rieskamp2006extending} and animals \citep*{shafir2002}.  Recently, \cite{Batley2016} estimated nested logit parameters to check for consistency with regularity (and various forms of stochastic transitivity) and found that parameter values consistent with violations of regularity provided the best fit.} Behaviorally, we can interpret $\eta_i >1$ as indicating complementarities among alternatives.\footnote{Relatedly, \cite*{ortoleva2019deliberate} suggests that regularity may be violated due to deliberate randomization between complementary lotteries.} In certain contexts, we may even anticipate $\eta_i >1$. For instance, \cite{Foubert2007} study the effects of ``product bundling'' and find a parameter greater than one, consistent with the effectiveness of bundling.\footnote{Indeed, in regard to whether $\eta_i$ should be less than or greater than one, \cite*{train1987demand} state that ``...the value of [$\eta_i$] indicates relative substitutability within and among nests, and neither possibility can be rule out \emph{a priori}.''} As NSC allows for violations of regularity, formally defined below, the NSC is not nested by RUM.  
\begin{defn}[Regularity]\label{reg} A stochastic choice function $p$ satisfies \textbf{regularity} if $p(x, A\cup y) \le p(x, A)$ for any $A \in \mathscr{A}$ and $x, y\in X$ with $x\in A$.
\end{defn}

Finally, note that when $\eta_i=1$ the nest value is exactly proportional to the sum of Luce utilities. If this proportionality happens for every nest $i$, the model reduces to Luce. In fact, in this case the Luce model has multiple NSC (and nested logit) representations with different partitions and identification of a unique nest structure is not possible. To rule this out, we say that an NSC $p$ with $(v, u, \{X_i\}^K_{i=1})$ is \emph{nondegenerate} if there is at most one nest where this proportionality occurs: there is at most one $i\le K$ such that for some $a\in X_i$,
\[\frac{\sum_{x\in A_i}u(x)}{v(A_i)}= \frac{u(a)}{v(a)}\text{ for any }A_i\subseteq X_i\text{ with }a\in A_i.\]This restriction rules out cases when $v$ is always proportional to the sum of Luce utilities. Further, the Luce model has a unique nondegenerate NSC representation in which there is a single nest, $X_1=X$.\footnote{Indeed, if there are $i, j$ such that
$\frac{\sum_{x\in A_i}u(x)}{v(A_i)}=\frac{u(a)}{v(a)}$ and $\frac{\sum_{y\in A_j}u(y)}{v(A_j)}=\frac{u(b)}{v(b)}$ for any $A_i\subseteq X_i, A_j\subseteq X_j$, $a\in A_i$, and $b\in A_j$, then $X_i\cup X_j$ should be treated as one nest. It is also not difficult to show that the set of degenerate NSC is measure zero with respect to the set of all NSC.} This nondegeneracy condition will be crucial for the unique identification of nests, but it is not required for the sufficiency part of  our characterization (Theorem 1).

\section{Behavioral Characterizations}\label{nestedchoice}

\subsection{The Similarity Effect and Revealed Similarity}\label{similarity}

Following the intuition behind the similarity effect, we introduce a notion of revealed similarity that will be essential to our analysis. Consider the effect of adding an alternative $x$ on the probabilities of choosing $a$ and $b$ from some menu $A$.  Adding $x$ might decrease these probabilities as it competes with $a$ and $b$.  If $x$ disproportionately affects one of them, say $a$ relative to $b$, this reveals that $a$ and $b$ are dissimilar. Conversely, if $x$ takes away from $a$ and $b$ proportionally, then this reveals that $a$ and $b$ are similar (symmetric) in menu $A$. We take a conservative approach and call two alternatives similar only if this is true for any menu $A$ (they are symmetric to all other alternatives). 

\begin{defn}For any alternatives $a, b\in X$, we say that $a$ and $b$ are \textbf{revealed categorically similar}, denoted by $a\sim_p b$, if $p$ satisfies IIA at $a, b$. We also say that $a$ and $b$ are \textbf{revealed categorically dissimilar} if $a\not\sim_p b$.
\end{defn}

The similarity effect is often defined using an exogenously given similarity relation. With our formal notion of revealed similarity, we may establish a fully behavioral definition of the similarity effect given $\sim_p$.\footnote{There are other ways to define similarity and other properties one might demand of a similarity relation. For instance, \cite{rubinstein1988similarity} studies similarity and choice under risk. In his paper, the similarity relation is reflexive and symmetric, like ours, but also must satisfy a form of betweenness with respect to objective attributes and violates transitivity, unlike ours. \cite{natenzon2018random} introduces a notion of comparative similarity based on absolute rather than relative choice frequencies. This similarity notion is not related to IIA and will not induce a partition structure on the set of alternatives.} Recalling the red bus/blue bus example, adding the blue bus had a larger effect on the red bus than on the train. Hence the blue bus ``takes more away'' from similar alternatives than from dissimilar alternatives.

\begin{defn} A stochastic choice function $p$ exhibits the \textbf{similarity effect} if for all $A \in \mathscr{A}$, $a, b\in A$, and $x\not\in A$, 
\[\text{ if }a\sim_p x\text{ and }b\not\sim_p x,\text{ then }\,\, \frac{p(a,A\cup x)}{p(b,A\cup x)} < \frac{p(a,A)}{p(b,A)} .\]  
\end{defn}
Intuitively, $x$ hurts the revealed categorically similar alternative $a$ more than a revealed categorically dissimilar alternative $b$. Since $a$ and $x$ are closer substitutes, $x$ competes more with $a$ than it does with $b$.

\subsection{Characterization of NSC}\label{CharNSC}

In order to introduce our axiom, we consider the general effect of introducing an alternative $x$ on the choice probabilities of two alternatives $a$ and $b$. IIA requires that the relative probability between $a$ and $b$ is \emph{always} independent of $x$. However, as the similarity effect suggests, (asymmetric) similarity between $x$ and $a, b$ might affect the relative probabilities. We therefore divide IIA into two logically independent axioms based on the revealed similarity between $x$ and $a, b$.

\begin{axm}[Independence of Symmetric Alternatives]\label{ISA} For any $A\in\mathscr{A}$, $a, b\in A$, and $x \notin A$,

\begin{equation*}
\begin{aligned}[c]
&\,a\sim_p x\text{ and }b\sim_p x\\
&\,\,\,\,\,\,\,\,\,\,\,\,\,\,\,\,\text{ or }\\
&\, a\not\sim_p x\text{ and }b\not\sim_p x
\end{aligned}
\qquad\Longrightarrow\qquad
\begin{aligned}[c]
\frac{p(a, A)}{p(b, A)}=\frac{p(a, A\cup x)}{p(b, A\cup x)}.
\end{aligned}
\end{equation*}
\end{axm}

\begin{axm}[Independence of Asymmetric Alternatives]\label{IAA} For any $A\in\mathscr{A}$, $a, b\in A$, and $x \notin A$,
\begin{equation*}
\begin{aligned}[c]
a\sim_p x\text{ and }b\not\sim_p x
\end{aligned}
\qquad\Longrightarrow\qquad
\begin{aligned}[c]
\frac{p(a, A)}{p(b, A)}=\frac{p(a, A\cup x)}{p(b, A\cup x)}. 
\end{aligned}
\end{equation*}
\end{axm}

The first axiom requires that the relative probability between $a$ and $b$ is independent of $x$ when $a$ and $b$ are revealed categorically (dis)similar to $x$. Intuitively, if $a$ and $b$ are symmetric from the perspective of $x$, then $x$ should symmetrically influence $a$ and $b$; it does not affect the relative probability between $a$ and $b$. The similarity effect directly contradicts the second axiom yet is unrelated to the first axiom.

\begin{obs} \emph{Luce's IIA is equivalent to the joint assumption of \textbf{\nameref{ISA}} and \textbf{\nameref{IAA}}. Moreover, the two axioms are independent.}
\end{obs}

We show in \autoref{NSCthm} that NSC is characterized by \nameref{ISA}, and thus NSC is precisely the generalization of Luce's model that accommodates the similarity effect.

\begin{thm}\label{NSCthm} Let $p$ be a stochastic choice function with at least three alternatives that are dissimilar to each other. Then $p$ satisfies \textbf{\nameref{ISA}} if and only if $p$ is a \textbf{nondegenerate NSC}.\end{thm}

\autoref{NSCthm} characterizes NSC when there are at least three nests; $a\not\sim_p b$, $b\not\sim_p c$, and $a\not\sim_p c$ for distinct alternatives.\footnote{When the assumption is violated (i.e., there are only two nests), we can still obtain the characterization result by modifying \nameref*{ISA}. In particular, we can impose a modification of Luce's (1959) Product Rule instead of the second part of \nameref*{ISA}. It is well known that IIA is equivalent to the Product Rule for menus with two alternatives (see Luce (1959)).}
While the proof is in the appendix, we discuss briefly how our axiom characterizes NSC. It should be apparent from the definition that $\sim_p$ is reflexive and symmetric. It turns out that the first part of \nameref{ISA} ($a\sim_p x$ and $b\sim_p x$) implies that $\sim_p$ is transitive.\footnote{More general notions of similarity may be intransitive (e.g., due to context dependence). Since we take a conservative definition of similarity, we find transitivity quite reasonable in our setting.  That is, by requiring  $\frac{p(a, A)}{p(b, A)}=\frac{p(a, \{a, b\})}{p(b, \{a, b\})}$ for any menu, we eliminate much of the context dependence. Further, transitivity of this revealed similarity relation is implicitly assumed in nested logit. See \autoref{Disc}.} Hence, transitivity of $\sim_p$ immediately generates a partition $X_1, \ldots, X_k$ of $X$ (or disjoint nests) such that any two alternatives in $X_i$ are revealed categorically similar.\footnote{Transitivity of $\sim_p$ is imposed in \cite{li2016associationistic}, which will be carefully discussed in \autoref{rellit}. A weak version of transitivity of $\sim_p$ is also used in \cite{echenique2018perception}.} However, by itself it imposes no particular structure on choice, nor does it establish a relationship between the partition and choices (except that IIA is satisfied within each nest).  The essential structure of NSC is captured by the second part of \nameref{ISA} ($a\nsim_p x$ and $b\nsim_p x$). Therefore, almost all of the proof is devoted to showing that the second part of \nameref{ISA} implies a nested choice structure consistent with this partition.

Lastly, we state the uniqueness properties of the NSC representation. The following proposition shows that the nest structure is unique, the nest utility $v$ is unique up to a positive scalar, and Luce's utility $u$ is unique up to a positive scalar at each nest.   

\begin{prps}[Uniqueness]\label{NSCuniqueprop} Suppose $p$ is a nondegenerate NSC with respect to $(v, u, \{X_i\}^K_{i=1})$ as well as to $(v', u', \{X'_i\}^{K'}_{i=1})$. Then $K=K'$ and $\{X_i\}^{K}_{i=1}=\{X'_i\}^{K'}_{i=1}$. Moreover, there is $(\alpha_1, \ldots ,\alpha_K, \delta)\in\mathds{R}^{K+1}_{++}$ such that $v'=\delta\, v$ and for any $x_i\in X_i$, $u'(x_i)=\alpha_i\,u(x_i)$.\end{prps}

\subsection{Characterizations of Nested Logit}\label{nested}

The most well-known special case of NSC is nested logit, which was specifically created to handle the similarity effect. The difference between NSC and nested logit is that the latter imposes a special structure on the nest values, $v(A\cap X_i)=\big(\sum_{a\in A\cap X_i} u(a)\big)^{\eta_i}$ with $\eta_i>0$, which has non-trivial behavioral consequences.

Despite its widespread use, nested logit has not been subject to careful axiomatic analysis in the way that other choice models have been. We provide two characterizations of nested logit that clarify the behavioral assumptions embedded in this model. The first characterization uses one additional axiom that imposes a menu independence condition on certain probability ratios.

\begin{axm}[Log Ratio Invariance]\label{LRI} For any $a, x\in X$ and $A, B\in\mathscr{A}$ such that $a\sim_p a'$ for all $a' \in A\cup B$,  
\[\frac{\log\Big(\frac{p(A,\, A\,\cup\, x)}{p(x,\, A\,\cup\, x)}\big{/}\frac{p(a,\, \{a,\, x\})}{p(x,\, \{a,\, x\})}\Big)}{\log\Big(\frac{p(A, \,A\,\cup\, a)}{p(a, \,A\,\cup\, a)}\Big)}=\frac{\log\Big(\frac{p(B,\, B\cup\, x)}{p(x,\, B\cup\, x)}\big{/}\frac{p(a,\, \{a,\, x\})}{p(x,\, \{a,\, x\})}\Big)}{\log\Big(\frac{p(B, \,B\,\cup\, a)}{p(a, \,B\,\cup\, a)}\Big)}.\]
\end{axm}

Log Ratio Invariance requires that the ratio $\log\!\big(\frac{p(A,\, A\,\cup\, x)}{p(x,\, A\,\cup\, x)}\big{/}\frac{p(a,\, \{a,\, x\})}{p(x,\, \{a,\, x\})}\big)$ and $\log\!\big(\frac{p(A, \,A\,\cup\, a)}{p(a, \,A\,\cup\, a)}\big)$ are proportional. 

\medskip
\begin{thm}\label{NL} A nondegenerate NSC satisfies \textbf{\nameref{LRI}} if and only if it is an \textbf{nested logit}.
\end{thm}

The explicit use of a functional form in Log Ratio Invariance allows us to establish testable implications for the functional form assumed in nested logit even with finite data. 

In the rest of this section, we discuss an alternative axiom that captures the essential features of nested logit without an explicit functional form and shows that is characterizes nested logit under a richer domain assumption. To state this axiom, first notice that an important behavioral property of nested logit, beyond its treatment of similarity (ISA), is that the probability of choosing a nest depends on the total attractiveness of the nest: $v(A\cap X_i)$ is increasing in $\sum_{a\in A\cap X_i} u(a)$. 

This behavior is characterized by a simple monotonicity property imposed among similar alternatives. Suppose that $A, A'\in\mathscr{A}$ are menus such that all alternatives in $A\cup A'$ are revealed similar. When $A$ is more attractive than $A'$, then alternatives in $A$ are always chosen more frequently than alternatives in $A'$ when they are compared with any other alternative $x$. More formally, $p(A, A\cup A')\ge p(A', A\cup A')$ implies $p(A, A\cup x) \ge p(A', A'\cup x)$ for any $x\not\in A\cup A'$. This can be viewed as an additional form of context independence, as it requires that there is no interaction between alternatives in $A\cup A'$ and $x$ which might create a choice frequency reversal.

Because nested logit involves a power function, it satisfies a stronger version of the monotonicity property above. In particular, the monotonicity property holds even in relative terms: if $A$ is relatively more attractive than $A'$ when they are compared to any other menus, $B$ and $B'$, then alternatives in $A$ will be chosen relatively more frequently than $A'$ when they are chosen against $x$. 
 
\begin{axm}[Relative Likelihood Independence]\label{RLI} For any $x\in X$ and $A, B, A', B'\in\mathscr{A}$ such that $a\sim_p a'$ for any $a, a'\in A\cup B\cup A'\cup B'$,
\[\frac{p(A, A\cup B)}{p(B, A\cup B)}\ge \frac{p(A', A'\cup B')}{p(B', A'\cup B')} \implies \frac{p(A, A\cup x)}{p(x, A\cup x)}\Big{/}\frac{p(B, B\cup x)}{p(x, B\cup x)}\ge \frac{p(A', A'\cup x)}{p(x, A'\cup x)}\Big{/}\frac{p(B', B'\cup x)}{p(x, B'\cup x)}.\]
\end{axm}

In our next result, we prove that Relative Likelihood Independence is a necessary condition for nested logit. Moreover, it implies that $v(A\cap X_i)$ is increasing in $\sum_{a\in A\cap X_i} u(a)$.

\begin{prps}\label{proposition2} Any nested logit satisfies \textbf{\nameref{RLI}}. Conversely, if a nondegenerate NSC with $(v, u, \{X_i\}^K_{i=1})$ satisfies Relative Likelihood Independence, then for each $i\le K$ there is an increasing function $f_i:\mathds{R}_{++}\to\mathds{R}_{++}$ such that $v(A)=f_i\big(\sum_{x\in A}u(x)\big)$ for any $A\subseteq X_i$. 
\end{prps}

While \nameref{RLI} is not sufficient for nested logit, this is essentially due to the limitations of finite data. Indeed, we show that Relative Likelihood Independence characterizes nested logit when the following richness condition is satisfied. 

\begin{axm}[Richness]\label{richness} For any $a\in X$ and $\rho\in (0, 1)$, there is $b\in X$ such that $a\sim_p b$ and $p(a, \{a, b\})=\rho$.  
\end{axm}

On its own, \nameref{richness} is relatively mild and simply ensures that there are alternatives for each utility value. However, under this condition \nameref{RLI} yields the well-known functional form of nested logit. Consequently, \nameref{RLI} captures all remaining behavioral features of nested logit.

\begin{thm}\label{nlthm} Any nondegenerate NSC that satisfies \textbf{\nameref{RLI}} and \textbf{\nameref{richness}} is a nested logit. 
\end{thm}

In applied settings, nested logit is often restricted to $\eta_i \in (0,1]$, as this is sufficient for it to be a RUM.  Since any RUM satisfies \nameref{reg}, a random utility nested logit must as well.  It is well known that \nameref{reg} is necessary but not sufficient for a model to be a RUM in general.  However, we show that \nameref{reg} is sufficient for the nested logit to be a RUM under \nameref{richness}.

\begin{prps}\label{runlprop} Any nested logit that satisfies \textbf{\nameref{reg}} and \textbf{\nameref{richness}} is a random utility nested logit.
\end{prps}

If $\eta_i > 1$, nested logit will violate regularity for some specifications of $u$.  Therefore \nameref{richness} strengthens the bite of \nameref{reg} and $\eta_i \in (0,1]$ is ensured.

\subsection{Beyond Nested Logit}\label{bnl}

It is a matter of folk knowledge that the aforementioned ``IIA within a nest and IIN'' properties serve as the behavioral underpinnings of nested logit. However, our results show that ``IIA within a nest and IIN'' (with an endogenous nest structure) in fact characterize NSC, and nested logit requires an additional property (\nameref{RLI}). In this subsection, we present three special cases of NSC that are distinct from nested logit. These examples illustrate some natural choice behaviors that \nameref{RLI} rules out, further clarifying the behavioral assumptions behind nested logit.

\subsubsection{Linear NSC}

One interesting example of NSC that is distinct from nested logit is the Linear NSC. In this example, the nest value is linear in total nest utility, in contrast to the power function used in nested logit. For each nest $i$, there exist parameters $\lambda_{i}\ge 0$ and $\nu_{i},$ so that 
\begin{equation}\label{mulsub} v(A\cap X_i) = \lambda_i \Big(\sum_{x \in A\cap X_i}u(x)\Big)+\nu_i.\end{equation}
In the Linear NSC, the attractiveness of a nest depends on both its instrumental utility, through $\lambda_i$, and an intrinsic ``category'' utility, through $\nu_i$. It turns out that the Linear NSC is a special case of both Elimination-by-Aspects (EBA) of \cite{tversky1972eba} and the Attribute Rule (AR) of \cite*{gul2014random}. Consequently, the Linear NSC is also a RUM.

\subsubsection{Menu-Dependent Substitutability}

In nested logit, the nest parameter $\eta_i$ captures substitutability of alternatives.  While the standard nested logit only allows for a single substitution parameter for each nest, NSC can accommodate \emph{menu-dependent substitutability}. For instance, consider the following example where substitutability depends on the size of the menu, capturing the idea that consumers may find it more difficult to distinguish between alternatives in larger option sets.

For each nest $i$, there exists a threshold $\tau_i \in \{1,\ldots,|X_i|\}$, and nest parameters $\eta_{i}, \tilde{\eta}_{i} > 0,$ so that \begin{equation}\label{mulsub} v(A\cap X_i) = \begin{cases} \left(\sum_{x \in A\cap X_i}u(x)\right)^{\eta_{i}} & \text{ if } |A\cap X_i| > \tau_i \\ \left(\sum_{x \in A\cap X_i}u(x)\right)^{\tilde{\eta}_{i}} & \text{ if } |A\cap X_i| \le \tau_i. \end{cases}\end{equation} If $1-\eta_{i} > 1-\tilde{\eta}_{i}$, this means the agent perceives fewer differences among alternatives as the nest becomes ``more represented.'' That is, when $ |A\cap X_i|$ exceeds $\tau_i$, alternatives are ``more substitutable.'' In this example, $\tau_i$ has a natural interpretation as the consumer's ``distinction capacity.'' Further, if $1-\eta_{i} > 0 > 1-\tilde{\eta}_{i}$, then whether the alternatives are complements or substitutes depends on the size of the nest. Lastly, when $\eta_i$ and  $\tilde{\eta}_{i}$ are both less than one and are sufficiently close, this example is also consistent with RUM.

\subsubsection{Attention and Spillover Effects}

The NSC also allows for cases where the nest value is not directly tied to alternative utility. We consider a particular example in which $v$ is determined by the ``salience'' of alternatives. For some function $S:X \ra \R_{++}$, 
\begin{equation}v(A\cap X_i)=\max_{x \in A\cap X_i}S(x).\end{equation} 

In this specification, the value of a nest is determined by the ``attractiveness'' or ``noticeability'' of its most salient alternative. 
To illustrate its behavioral implications, suppose there are three alternatives, $X=\{x,y,z\}$, with nests $X_1=\{x,z\}$ and $X_2=\{y\}$. When $z$ is highly salient but low utility, $S(z) u(x) > S(x)[u(x)+u(z)]$, then $\frac{p(x,\{x,y,z\})}{p(y,\{x,y,z\})}  > \frac{p(x,\{x,y\})}{p(y,\{x,y\})}$. Examples of such $z$ include brands offering a high-end good with a high price to attract attention, expecting all consumers to purchase their ``moderate'' offering $x$.  When the value of $S(z)$ is large enough relative to the value of $u(z)$, this may induce a violation of regularity.  Similar examples can generate ``spillover'' effects. For example, one successful or attractive product may funnel attention to others, causing demand spillover. This is the traditional rationale behind the use of ``loss-leaders" (\cite{lal1994}) or ``attention-grabbers'' (\cite{eliaz2011attention}).

\section{Revealed Similarity and its Extensions}\label{Disc}

We say that two alternatives are revealed categorically similar if IIA is satisfied between them at all menus. Requiring this eliminates the menu dependence of similarity, and so our notion of revealed similarity captures a form of absolute or fundamental similarity. Consequently, similarity is symmetric and, under \nameref{ISA}, transitive. One drawback is that this notion does not allow statements about comparative similarity; two alternatives are similar or not.  Additionally, in some cases impressions of similarity may be context-dependent.\footnote{There is a sense in which our notion is somewhat moderate. Consider Debreu's red bus/blue bus example. In this case, the similar alternatives (buses) are in fact identical, often called duplicates (or in some cases replicas). Duplicates are not merely similar alternatives; they are similar and \emph{provide the exact same utility value.} For example, in \cite{gul2014random} the use of duplicates is essential to their characterization of the Attribute Rule. Formally, $x$ and $y$ are duplicates if $p(a, A\cup x)=p(a, A\cup y)$ for any $A$ and $a\in A$. However, our notion of similarity is not tied to utility. A commuter may regard all buses as similar (i.e., they belong to the same nest), yet nothing in our model restricts an agent from exhibiting a preference over different buses (e.g., because some bus routes may be faster or cheaper than others).} Because of these apparent limitations, we consider two ways in which to extend NSC to accommodate more complex notions of similarity. 

The first extension of NSC relaxes the requirement that an alternative must belong to a single nest.  In section \ref{CNL}, we consider the (unrestricted) \emph{cross-nested logit} \citep{vovsha1997application} and the (more general) generalized nested logit (\cite{wen2001gnl}), which allow for each alternative to be ``allocated'' across several nests. While overlapping nests allows for the most flexible notion of similarity, these models have no testable implications if the nests are not known \emph{a priori} and parameter values are unrestricted. Thus we demonstrate an important trade-off between nested and cross-nested models.

The second extension of NSC allows for ``intermediate nests." These intermediate nests are often visually represented through a multi-level decision tree. Within this structure, we can allow for a more nuanced notion of similarity through the introduction of a second similarity relation that is conceptually related to our core similarity notion. This secondary relation captures ``context-dependent'' similarity and allows for comparative statements.  We provide an axiomatic characterization of this model (\autoref{3NSCthm}) in appendix \ref{3nsc}.\footnote{Just as our similarity relation identifies endogenous nests, this secondary relation identifies endogenous, intermediate nests. Thus, \autoref{3NSCthm} shows that we may identify an endogenous tree structure.}

\subsection{Overlapping Nests}\label{CNL}

In NSC, each alternative belongs to one, and only one, nest. This feature of NSC places restrictions on the similarity relation. Because of these restrictions, in some settings, applied researchers have proposed allowing alternatives to exist in multiple nests. This leads to a class of models known as ``cross-nested logits'' (see \cite{vovsha1997application}, \cite{ben1999discrete}, \cite{wen2001gnl}, \cite{papola2004cnl}, and \cite{bierlaire2006cnl}).  In the cross-nested logit and the generalized nested logit, each alternative is allocated among the various nests. This allocation is specified with a vector of weights, one for each alternative, which describes to what extent an alternative belongs to each nest. 

We show that any stochastic choice rule $p$ can be rationalized by some unrestricted cross-nested logit. That is, for any $p$, there exist some nesting structure, $X_1, \ldots, X_K$, allocations to these nests, $(\alpha^k_x)^K_{k=1}$, and utilities so that the resulting unrestricted cross-nested logit generates identical choice frequencies. Hence, unlike the nested logit and Luce models, there can be no behavioral characterization of the unrestricted cross-nested logit. The only testable implications of the model are due to the analysts' assumptions about alternative categorization and parameter restrictions.

\begin{defn}[Generalized Nested Logit]\label{cross-nested} A stochastic choice function $p$ is an unrestricted \textbf{generalized nested logit} if there is a collection of subsets $X_1, \ldots, X_K$ of $X$ and a vector $(\alpha^k_x)^K_{k=1}\in \mathds{R}^{K}_{+}$ with $\sum^K_{k=1} \alpha^k_x=1$ for each $x\in X$ such that $x\not\in X_k$ iff $\alpha^k_x=0$, a utility function $u:X\to \mathds{R}_{++}$, and parameters $(\lambda_k)^K_{k=1}\in\mathds{R}^K_{++}$ such that for any $A\in\mathscr{A}$ and $x\in A$, 
\begin{equation}
p(x, A)=\sum_{k: x\in A\cap X_k}\frac{\big(\alpha^k_x\, u(x)\big)^\frac{1}{\lambda_k}}{\sum_{y\in A\cap X_k}\big(\alpha^k_y\, u(y)\big)^\frac{1}{\lambda_k}}\cdot\frac{\Big(\sum_{y\in A\cap X_k}\big(\alpha^k_y\, u(y)\big)^\frac{1}{\lambda_k}\Big)^{\lambda_k}}{\sum_{l:A\cap X_l\neq\emptyset}\Big(\sum_{z\in A\cap X_l}\big(\alpha^l_z\,u(z)\big)^\frac{1}{\lambda_l}\Big)^{\lambda_l}}.\end{equation}
Moreover, we say that $p$ is an \textbf{unrestricted cross-nested logit} if $\lambda_k=\lambda_{k'}$ for any $k, k'\le K$.  
\end{defn}

\begin{thm}\label{CNLthm} Every stochastic choice function is an \textbf{unrestricted cross-nested logit}.\end{thm}

\begin{cor}\label{GNLthm} Every stochastic choice function is an \textbf{unrestricted generalized nested logit}.\end{cor}

Our result relies on the key insight that the cross-nested logit is behaviorally equivalent to a form of menu-dependent utility. We first prove this equivalence as \autoref{cnl-md} and show how we can go from menu-dependent utility to weighted allocations and back.  This equivalence between the  cross-nested logit and menu-dependent utility allows us to reduce the problem of finding weights to the problem of finding menu-dependent utility values for each menu that satisfy the cross-nested logit equation. The bulk of the proof is dedicated to showing that the existence of these menu-dependent utilities is equivalent to the existence of a fixed point for some self-map. The proof is completed by applying Brouwer's fixed point theorem. 

This result precisely shows the trade-offs between using nested logit and related models: either accept a restrictive form of similarity or impose assumptions on nest structure and model parameters. As we mentioned previously, further assumptions on parameters or nest structure may lead to testable restrictions. In the literature, similar to nested logit, it is commonly assumed that $\lambda\le 1$, since this is sufficient for cross-nested logit to be RUM. As with our handling of nested logit, we refer to such specifications as the \emph{random utility cross-nested logit}. Note that our result shows that this restriction is not necessary for consistency with RUM; By \autoref{CNLthm}, every RUM has an unrestricted cross-nested logit representation with $\lambda >1$.

In any case, a random utility cross-nested logit must have, at least, the same testable restrictions as RUM. However, our result suggests that random utility cross-nested logit may not have any testable restrictions beyond RUM. In fact, although it does not prove our hypothesis, \cite*{fosgerau2013choice} prove that the set of random utility cross-nested logit models is dense in the set of RUMs. Note that our \autoref{CNLthm} is quite different from the result of \cite{fosgerau2013choice} for the following reasons: (i) we prove an exact result while they prove an approximation result, (ii) they focus on random utility cross-nested logit, and (iii) our proof techniques are completely different because their proof relies on the properties of the CDF for GEV distributions while we use Brouwer's fixed point theorem.

\section{Identifying Nests}\label{IDnest}

In most applications of nested logit to market data, researchers assume nests based on knowledge of alternative attributes. This is potentially problematic, as in many environments there are many plausible structures. When studying vehicle choice, the researcher might construct nests based on brand, body type (e.g., sedan vs. truck), or country of origin.\footnote{A common approach to this type of problem is to utilize multiple levels of nesting (which we characterize in appendix \ref{3nsc}). Even under this approach, the order of the levels matters.}  In other environments, classification may be subjective. When studying choice over apartments, nests might depend on both observable attributes and a myriad of unobservables: subjective impressions of neighborhoods, proximity to landmarks, or a host of other features.  If the nest structure is misspecified, this may lead to biased conclusions regarding substitutability of goods and systematically inaccurate predictions.\footnote{\cite{greene2003econometric} provides an excellent summary of this issue: ``\emph{To specify the nested logit model, it is necessary to partition the choice set into branches [nests]. Sometimes there will be a natural partition ... In other instances, however, the partitioning of the choice set is ad hoc and leads to the troubling possibility that the results might be dependent on the branches so defined. ... There is no well-defined testing procedure for discriminating among tree structures, which is a problematic aspect of the model.}"}

We show in \autoref{distmin} that the true (unobserved) nest structure can be identified from the data by solving a minimization problem.  Any potential nest structure has implications for when IIA may and may not be violated between alternatives.  For a hypothesized nest structure, $\mathcal{Y}$, we propose a measure of the total magnitude of IIA violations within and across the proposed nests, $D(\mathcal{Y})$. We show that the true nest is a minimizer of $D$ and that it will be the unique minimizer of $D$ under a mild identification assumption (\autoref{NMPprop}).  In cases where the researcher has several potential nest structures in mind, such as in vehicle choice, our procedure for nest identification could be useful for nest selection. The researcher can calculate $D$ for the particular nests in mind and select the one that best fits.

 Because the number of possible nests grows exponentially as the number of alternatives increases, the full minimization problem becomes intractable. However, this issue can be managed due to insights from our similarity relation; by \autoref{Reduceprop}, one only needs to check nest structures that are consistent with an empirical approximation of $\sim_p.$  Note that in finite data sets it is unlikely that IIA will hold between any alternatives (e.g., since we observe a finite sample from the true distribution). However, one can measure the magnitude of the the IIA violation between $a, b$ across various menus in the data.  If this magnitude is below some threshold $\epsilon$, then we conclude that $a$ and $b$ are approximately similar: $a \sim_{\epsilon} b$.  When $\sim_{\epsilon}$ is transitive, then there are at most $|X|$ potential nests that we need to check, as stated in \autoref{NestReduceprop}.

\subsection{Identifying Nests by Distance Minimization}\label{distmin}

To analyze the problem of nest identification, we consider a data set denoted $\mathcal{O}=\{A, \{p_{t}(\cdot, A)\}^{N_A}_{t=1}\}_{A\in\mathscr{A}}$, where $N_A$ is the number of observations of menu $A$ and $p_{t}(a, A)=1$ means that $a$ was chosen from $A$ at observation $t\le N_A$. We also require $\sum_{x\in A}p_{t}(x, A)=1$ for each $A$, so that $p_{t}(a, A)=1$ implies $p_{t}(b, A)=0$ for any $b\in A\setminus\{a\}$. Note that we may always write
\[p_{t}(a, A)=\overline{p}(a, A)+\epsilon_{t, a, A},\]
where $\overline{p}(a, A)$ is the probability that $a$ is chosen from $A$ according to the NSC with $(v, u, \{X_i\}^K_{i=1})$. Then, the observed choice frequency of $a$ from $A$ in $\mathcal{O}$ is 
\[p(a, A)\equiv\frac{\sum^{N_A}_{t=1}p_t(a, A)}{N_A}=\overline{p}(a, A)+\epsilon_{a, A}\text{ where }\epsilon_{a, A}\equiv\frac{\sum^{N_A}_{t=1}\epsilon_{t, a, A}}{N_A}.\]
We assume that $\{p_t(\cdot, A)\}^{N_A}_{t=1}$ are independently drawn according to $\overline{p}(\cdot, A)$. Then, by the classical Glivenko-Cantelli theorem, $\epsilon_{a, A}\xrightarrow{a.s.} 0$.\footnote{All convergence statements in this paper are with respect to $N^*\to\infty$ where $N^*=\min_{A\in\mathscr{A}} N_A$.} For notational simplicity, we write \[r_A(A', B')\equiv\frac{p(A', A)}{p(B', A)}\text{ and }\bar{r}_A(A', B')\equiv \frac{\overline{p}(A', A)}{\overline{p}(B', A)}\text{ for any }A, A', B'\in\mathscr{A}.\] Finally, let $\mathscr{X}$ denote the set of all partitions of $X$ and $\mathcal{X}^*$ denote the true partition $\{X_i\}^K_{i=1}$.

Consider the following minimization problem. 
\begin{equation}\label{NMP}\tag{\textbf{NMP}}\min_{\mathcal{Y}\in\mathscr{X}}D(\mathcal{Y})=D_1(\mathcal{Y})+D_2(\mathcal{Y}),\end{equation}
\begin{equation}\label{D1eq}D_1(\mathcal{Y})=\frac{\sum_{Y\in \mathcal{Y}}\sum_{A, B\in\mathscr{A}, a, b\in A\cap B\cap Y}\Big(\log\big(r_A(a, b)\big)-\log\big(r_B(a, b)\big)\Big)^2}{\sum_{Y\in\mathcal{Y}}|\{(A, B, a, b)| a, b\in A\cap B\cap Y\}|},\end{equation}
\begin{equation}\label{D2eq}D_2(\mathcal{Y})=\frac{\sum_{Y, Y'\in\mathcal{Y}} \sum_{A, B\in\mathscr{A}: A\cap Y=B\cap Y,\, A\cap Y'=B\cap Y'}\Big(\log\big(r_A(Y, Y')\big)-\log\big(r_B(Y, Y')\big)\Big)^2}{\sum_{Y, Y'\in\mathcal{Y}} |\{(A, B)| A\cap Y=B\cap Y,\,A\cap Y'=B\cap Y'\}|}.\end{equation}

Intuitively, $D_1(\mathcal{Y})$ measures the degree to which the data violates IIA among elements in the same nest in $\mathcal{Y}$, while $D_2(\mathcal{Y})$ measures the degree to which the data violates IIA among different nests in $\mathcal{Y}$.  These measures are motivated by our axiom \nameref{ISA}: $D_1$ measures the extent to which the first part of \nameref{ISA} holds, and $D_2$ measures the extent to which the second part of \nameref{ISA} holds. 

Similarly, let us define loss functions $D^*, D^*_1, D^*_2$ when there is no observational noise; these are defined by replacing $p$ with $\bar{p}$ in Equations \ref{D1eq}-\ref{D2eq}. Moreover, let $\hat{\mathcal{X}}=\arg\min_{\mathcal{Y}\in\mathscr{X}}D(\mathcal{Y})$. Note that $\hat{\mathcal{X}}$ is an M-estimator (\cite{takeshi1985advanced}). Hence, by standard results (\cite{newey1994large}), $\hat{\mathcal{X}}$ is a strongly consistent estimator of $\mathcal{X}^*$ if $\mathcal{X}^*$ is the unique minimizer of $D^*$. Indeed, $\mathcal{X}^*$ is a minimizer of $D^*$ since $D^*(\mathcal{X}^*)=0$. It turns out that it is the unique minimizer under the following identification assumption.

\begin{ass}\label{IdentifyAss}For all subsets $A_i\subset X_i$ and $A_j\subseteq X_j$, there are menus $A, B\in\mathscr{A}$ such that $\bar{r}_A(A_i, A_j)\neq \bar{r}_B(A_i, A_j)$, $A\cap A_i=B\cap A_i$, and $A\cap Y_j=B\cap Y_j$.
\end{ass}

We now can state our strong consistency result.

\begin{prps}\label{NMPprop} $\hat{\mathcal{X}}\xrightarrow{a.s.}\mathcal{X}^*$ under \autoref{IdentifyAss}.
\end{prps}

\autoref{NMPprop} shows that the true nest structure can be found by solving \ref{NMP}. The intuition behind the result is as follows. As we see in our axiomatization, $a\sim_p b$ if and only if $a, b\in X_i$ for some $i$. Hence, IIA is satisfied between $a$ and $b$ when $a, b\in X_i$ and IIA is violated at least once between $a$ and $b$ when $a\in X_i$ and $b\in X_j$. Hence, the distance $\sum_{A, B\in\mathscr{A}, a, b\in A\cap B\cap Y}\Big(\log\big(r_A(a, b)\big)-\log\big(r_B(a, b)\big)\Big)^2$ between $a$ and $b$ is smaller whenever $a, b\in X_i$. Hence, minimizing $D_1(\mathcal{Y})$ helps us to correctly identify that elements from different nests are in fact from different nests. 

However, it is important to notice that $D_1(\mathcal{Y})$ alone is not sufficient to identify $\mathcal{X}^*$. For example, suppose $X=\{a_1, \ldots, a_5\}$ and $\mathcal{X}^*$ is given by $X_1=\{a_1, a_2, a_3\}$ and $X_2=\{a_4, a_5\}$. Since the data provide a noisy measure of $\bar{p}$, it is possible that $D_1$ is minimized at some finer partition, say $Y_1=\{a_1\}$, $Y_2=\{a_2, a_3\}$, and $Y_3=\{a_4, a_5\}$. Note that $\mathcal{Y}$ splits $X_1$, and since $D_1$ measures IIA violations within nests, $D_1(\mathcal{Y})\le D_1(\mathcal{X}^*)$ because  $\mathcal{Y}$ never combines two elements from different nests into the same nest (i.e., it is finer than $\mathcal{X}^*$). 

This example illustrates a potential problem. $D_1$ by itself tends to select finer partitions (it wants to create ``too many nests''). The second term, $D_2$, corrects this problem. If $\mathcal{Y}$ were the true nest structure, our axiomatization (i.e., the second half of \nameref{ISA}) requires that the relative likelihoods between alternatives in $Y_2$ (for instance, $a_2$) and alternatives in $Y_3$  (for instance, $a_4$) are unaffected by the presence of $a_1$. Accordingly, $\mathcal{Y}$ is penalized by $D_2$ if introducing $a_1$ changes the relative likelihoods between alternatives in $Y_2$ and $Y_3$. Importantly, since the true nest structure, $\mathcal{X}^*$, groups $a_1$ with $a_2$ and $a_3$, $\mathcal{X}^*$ will not be penalized, and so $D_2(\mathcal{Y})> D_2(\mathcal{X}^*)$ almost surely.\footnote{We say \emph{$Z_n>Z'_n$ almost surely} when there is $N$ such that Pr$(Z_n>Z'_n)=1$ for any $n>N$.} Thus $D_2(\cdot)$ enables us to correctly conclude that $a_1$ and $a_2$ do in fact belong to the same nest.

Notice that solving \ref{NMP} is quite different from the typical exercise of selecting a nest structure in the literature. In a typical nested logit estimation, a researcher assumes a nest structure and then runs a maximum likelihood (ML) estimation to identify model parameters. To compare different nest structures, the researcher has to run a full ML estimation for each nest structure. Hence, it is computationally expensive to compare many different nest structures. However, our \ref{NMP} provides a data-driven way to compare different nest structures without estimating the full parametric model.  Moreover, \ref{NMP} does not rely on the functional form of nested logit, since it applies to any NSC.

Finally, note that $\mathcal{X}^*$ is a minimizer of $D$ without any further assumptions; our identifying \autoref{IdentifyAss} is only required to ensure that  $\mathcal{X}^*$ is the \emph{unique} minimizer. Consequently, when \autoref{IdentifyAss} is violated, $\mathcal{X}^*$ will always be contained in the set of minimizers of $D$. This suggests that $D$ may still be used for nest selection and that solving \ref{NMP} can facilitate identification of the true nest structure.

\subsection{Revealed Similarity and Nest Selection}\label{similarityselection}

There is a practical concern with directly applying \autoref{NMPprop} to identify the nest structure because $|\mathscr{X}|$ grows exponentially as $|X|$ increases.\footnote{Unlike the standard method of estimating nested logit, it is not computationally expensive to solve \ref{NMP} by going through all possible partitions of $X$ when $|X|$ is small. For instance, there are 877 different partitions when $|X|=7$. Indeed, many papers in the literature study situations with relatively few alternatives (e.g., transportation modes or cellphone providers), and \autoref{NMPprop} can be applied to these situations directly.} Therefore, we further refine our result by showing that we only need to compare at most $|X|$ different partitions, rather than $|\mathscr{X}|$. This dramatically reduces the number of calculations; comparing $|X|$ different partitions is computationally inexpensive even when $X$ contains hundreds of alternatives. To establish this result, we introduce the following measure of IIA violations. For any $a, b\in X$, let 
\begin{equation}\label{sime}d(a, b)\equiv \frac{\sum_{A, B\in\mathscr{A}: a, b\in A\cap B}\Big(\log\big(r_A(a, b)\big)-\log\big(r_B(a, b)\big)\Big)^2}{|\{(A, B, a, b)| a, b\in A\cap B\}|}.\end{equation}

The value of $d(a, b)$ captures the total ``distance'' between $a$ and $b$, in terms of IIA violations in the data $\mathcal{O}$.  Consistent with our axiomatization, and the intuition behind $D_1$, the value of $d(a,b)$ is smaller when $a$ and $b$ are from the same nest than when they are from different nests. While conceptually similar to $D_1$, note that it is defined over the alternatives, not on nest structures. This crucial distinction allows us to use $d$ to narrow our candidate nests purely based on the data.
 
\begin{prps}\label{Reduceprop} Under Assumption 1, there are $\epsilon^*, \bar{N}>0$ such that for any $N^*>\bar{N}$,
\[\max_{i}\max_{a, b\in X_i} d(a, b)<\epsilon^*<\min_{i<j}\min_{a'\in X_i, b'\in X_j} d(a', b')\,\text{ with probability one.}\]
\end{prps}

\autoref{Reduceprop} shows that for large enough $N^*$, there exists a ``separating" threshold that correctly identifies whether two alternatives belong, or do not belong, to the same nest. If the researcher knows $\epsilon^*$, then identifying the nest structure is a straightforward task due to this result. But when $\epsilon^*$ is unknown, \autoref{Reduceprop} is not sufficient to identify the nest structure. However, the insights provided by \autoref{Reduceprop} allow us to show that in order to identify the correct nest structure for any NSC, only $|X|$ different partitions are worth considering. In fact, we will explicitly construct the set of partitions that need to be considered using $d$ and prove that this set contains the true nest $\mathcal{X}^*$.

In order to construct the set of relevant partitions, we introduce the following ``approximately similar" relation: for any $\epsilon\ge 0$ and $a, b\in X$, let $a\sim_\epsilon b$ if $d(a, b)<\epsilon$. When $\sim_\epsilon$ is transitive, let $\mathcal{X}_\epsilon\equiv X/\sim_{\epsilon}$, which is the partition of $X$ such that for any $A\in\mathcal{X}_\epsilon$, $a\in A$, and $b\in X$, $a\sim_{\epsilon} b$ if and only if $b\in A$.  

Since we have finite data, if $\epsilon$ and $\epsilon'$ are close enough, they will result in the same relation ($\sim_{\epsilon}=\sim_{\epsilon'}$), except for certain knife-edge cases (which happens at most $|X|$ times).  Notice that for smaller $\epsilon$, we are ``more discriminating'' in declaring similarity and this results in a finer partition. For larger $\epsilon$, we are ``less discriminating'' in declaring similarity and this results in a coarser partition. Let $\overline{\epsilon}\equiv \max_{a, b\in X} d(a, b)$, the maximal distance calculated in the data. Then for any $\epsilon > \overline{\epsilon}$, the resulting relation $\sim_{\epsilon}$ is complete, which reduces to the Luce model ($\epsilon=0$ also gives the Luce model). Consequently, we never need to use an $\epsilon$ above $\overline{\epsilon}$. Because of these two key features of $\sim_{\epsilon}$, it turns out that the set $\mathscr{X}^*\equiv\{\mathcal{X}_\epsilon\}_{\epsilon\in [0, \overline{\epsilon}]}$ is the desired collection of partitions and contains at most $|X|$ different elements. 

\begin{prps}\label{NestReduceprop} $|\mathscr{X}^*|\le |X|$. \end{prps}

Combining Propositions \ref{NMPprop}-\ref{NestReduceprop}, we can immediately show that $\mathscr{X}^*$ contains the true partition structure, and it can be found by solving \ref{NMP}, as formalized below. Let $\hat{\mathcal{X}}^*=\arg\min_{\mathcal{Y}\in\mathscr{X}^*}D(\mathcal{Y})$.

\begin{cor} $\hat{\mathcal{X}}^*\xrightarrow{a.s.}\mathcal{X}^*$ under Assumption 1.\end{cor}

The minimization problem \ref{NMP} is not computationally demanding since $|\mathscr{X}^*|\le |X|$. Hence, we can find the true nest structure even if there are many products. In practice, computing $\mathscr{X}^*$ from choice frequencies is quite simple. First note that any partition of $X$ can be represented by an $|X|\times |X|$ matrix $M$ such that $M_{x, y}=1$ when $x$ and $y$ are from the same nest and $0$ otherwise. Hence, to compute $\mathscr{X}^*$, we follow the following steps: 
\begin{itemize}
\item[1)] Calculate $d(a, b)$ for each pair $(a, b)$ in $X \times X$;
\item[2)] For each $(a, b)$, construct the $|X|\times |X|$ matrix $M^{ab}$ such that $M^{ab}_{x, y}=1$ if $d(x, y)<d(a, b)$ and $0$ otherwise; 
\item[3)] Include the matrix $M^{ab}$ in $\mathscr{X}^*$ if for any $x, y, z\in X$, $M^{ab}_{x, z}=1$ whenever $M^{ab}_{x, y}=M^{ab}_{y, z}=1$. \end{itemize}

The first step determines the collection of relevant thresholds from the data to construct candidate relations $\sim_{\epsilon}$. The second step generates $|X|(|X|-1)/2$ matrices, which represent the similarity thresholds found in the previous step. The third step reduces the number to $|X|$, as we prove in \autoref{NestReduceprop}, since $\sim_{\epsilon}$ must be transitive.

\subsection{Identification from Simulations}

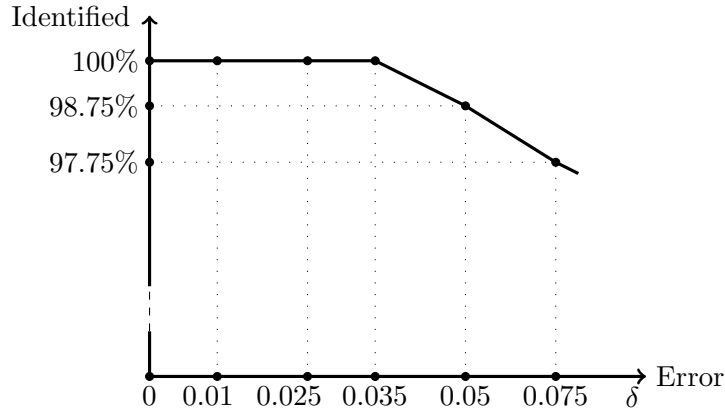
\begin{figure}
\centering
\begin{tikzpicture}[x=0.60cm, y=0.60cm][domain=0:1, range=0:1, scale=3/4, thick]%paper
\usetikzlibrary{intersections}
\setlength{\unitlength}{0.5 cm}
%utility space

\draw[->][very thick] (0,0)--(11, 0) ;
\draw[->][very thick] (0,2)--(0, 8); 

\draw[dashed] (0,1)--(0, 2); 

\draw[very thick] (0,0)--(0, 1); 

%\draw[very thick] (0,0)--(8,8); 

\coordinate[label=right: $\text{Error}$] (wf) at (11, 0);
\coordinate[label=below: $\text{$\delta$}$] (wf) at (10.7, 0);
\coordinate[label=left: $\text{Identified}$] (wf) at (-0.2, 8);

\coordinate[label=left: ${100\%}$] (wf) at (0, 7);

\filldraw (0,7) circle (1.5pt);

\coordinate[label=left: ${98.75\%}$] (wf) at (0, 6);

\filldraw (0,6) circle (1.5pt);

\coordinate[label=left: ${97.75\%}$] (wf) at (0, 4.75);

\filldraw (0,4.75) circle (1.5pt);

%\coordinate[label=left: ${95\%}$] (wf) at (0, 4);

%\filldraw (0,4) circle (1.5pt);

%\filldraw (7,0) circle (2pt);
%\filldraw (0,7) circle (2pt);

\draw[very thick] (0,7)--(1.5,7); 

\filldraw (1.5,7) circle (1.5pt);

\draw[very thick] (1.5,7)--(3.5,7); 

\filldraw (3.5,7) circle (1.5pt);

\draw[very thick] (3.5,7)--(5,7); 

\filldraw (5,7) circle (1.5pt);

\draw[very thick] (5,7)--(7,6); 

\filldraw (7, 6) circle (1.5pt);

\draw[very thick] (7,6)--(9,4.75); 

\filldraw (9, 4.75) circle (1.5pt);

\draw[very thick] (9,4.75) -- (9.5, 4.5); 

%\draw[blue][very thick] (7.5,5)--(10,4); 

%\filldraw (10, 4) circle (1.5pt);

%\coordinate[label=above: $\text{Budget Line}$] (wf) at (7, 2);
%\draw[->] (6.5,2)--(6,1.1); 

%\filldraw[opacity=0.15] (0, 0) -- (0,7) -- (7, 0) -- cycle;

%\draw[loosely dotted] (6,0)--(6,2);
%\draw[loosely dotted] (4,0)--(4,4);
%\draw[loosely dotted] (2,0)--(2,6);
%\draw[loosely dotted] (0,2)--(6,2);
%\draw[loosely dotted] (0,4)--(4,4);
%\draw[loosely dotted] (0,6)--(2,6);

%\filldraw[] (4,4) circle (3pt);
%\filldraw[] (4,2) circle (2pt);
%\coordinate[label=below:$x_1$] (wf) at (6, 0);
%\coordinate[label=below:$y_1$] (wf) at (4, 0);
%\coordinate[label=below:$z_1$] (wf) at (2, 0);
%\coordinate[label=left:$x_2$] (wf) at (0, 2);
%\coordinate[label=left:$y_2$] (wf) at (0, 4);
%\coordinate[label=left:$z_2$] (wf) at (0, 6);
%\coordinate[label=left:$X_2$] (wf) at (0, 7.5);

%\coordinate[label=above: $\text{Budget Set}$] (wf) at (1.8, 2);

%\filldraw[blue] (3.5,3.5) circle (3pt);
%\coordinate[label=above: $\text{Optimal Consumption}$] (wf) at (7.5, 3.5);
%\draw[->] (4.5,3.8)--(3.7,3.6); 

%\draw[very thick][dotted] (3.5,0)--(3.5, 3.5) ;
%\draw[very thick][dotted] (0,3.5)--(3.5, 3.5); 

\coordinate[label=below: ${0}$] (wf) at (0, 0);

\filldraw (0,0) circle (1.5pt);

\draw[loosely dotted] (1.5,0)--(1.5,7); 

\coordinate[label=below: ${0.01}$] (wf) at (1.3, 0);

\filldraw (1.5,0) circle (1.5pt);

\coordinate[label=below: ${0.025}$] (wf) at (3.1, 0);

\filldraw (3.5,0) circle (1.5pt);

\draw[loosely dotted] (3.5,0)--(3.5,7); 

\coordinate[label=below: ${0.035}$] (wf) at (5, 0);

\filldraw (5,0) circle (1.5pt);

\draw[loosely dotted] (5,0)--(5,7);

\filldraw (7,0) circle (1.5pt);

\draw[loosely dotted] (7,0)--(7,6); 

\draw[loosely dotted] (0,6)--(7,6); 

\coordinate[label=below: ${0.05}$] (wf) at (7, 0);

\coordinate[label=below: ${0.075}$] (wf) at (9, 0);

\filldraw (9,0) circle (1.5pt);

\draw[loosely dotted] (9,0)--(9,4.75);

\draw[loosely dotted] (0,4.75)--(9,4.75);

%\coordinate[label=below: ${0.1}$] (wf) at (10, 0);

%\filldraw (10,0) circle (1.5pt);

%\draw[loosely dotted] (10,0)--(10,4); 

%\coordinate[label=below: ${\bar{\delta}}$] (wf) at (4, 0);

%\filldraw (4,0) circle (1.5pt);

%\draw (4,0)--(4,7); 

\end{tikzpicture}
\caption{Percentage of trials with correctly identified nest structure with error $\delta$.}
\label{sim}
\end{figure}

To illustrate our algorithm and our theoretical result on identification, we ran the following simulation with six alternatives. We assumed that the true nest structure is $X_1=\{x_1, x_2, x_3\}$ and $X_2=\{x_4, x_5, x_6\}$, with $X=X_1\cup X_2$, and calculated the fraction of trials in which our procedure correctly identified the nest structure. To do so, we randomly generated values for $u$ and $v$ and calculated $\overline{p}$, which is the NSC given by $(v, u, \{X_1, X_2\})$. To introduce sampling error, we drew independent errors from a uniform distribution $U[0, \delta]$ and perturbed $\overline{p}$.\footnote{Specifically, for each menu $A$ and each simulation trial $t$, we independently draw errors $\{\zeta_{a, A,t}\}_{a\in A}$ from $U[0, \delta]$ and construct $\overline{p}^t(\cdot, A)$ as follows: $\overline{p}^t(a, A)=\frac{\overline{p}^t(a, A)+\zeta_{a, A, t}}{\sum_{b\in A}\overline{p}^t(b, A)+\zeta_{b, A, t}}$ for each $a\in A$.} As shown by \autoref{NMPprop}, when $\delta$ is small enough, the true nest structure will be identified correctly. This was confirmed by our simulation. 

We considered six different values for $\delta$ ($\{0, 0.01, 0.025, 0.035, 0.05, 0.075\}$) and ran a total of 2400 trials, the results of which are summarized in \autoref{sim}. For $\delta \in \{0, 0.01, 0.025, 0.035\}$, the nest structure was correctly identified in all trials. For $\delta=0.05$ ($0.075$), the nest structure was correctly identified 395 (391) times out of 400 trials. In other words, in line with our theoretical result (\autoref{NMPprop}), when error is relatively small (e.g., $\delta \le 0.035$) the true nest is recovered $100\%$ of the time. Even for relatively large errors (e.g., $\delta \ge 0.05$), we recover the true nests over $97\%$ of the time.

\section{Related Literature}\label{rellit}

The main contributions of our paper are the characterizations of nested logit and NSC. While nested logit is the most commonly applied model that deals with the similarity effect, many other models have been proposed. Two prominent such models are Elimination-By-Aspects (EBA) of \cite{tversky1972eba} and the Attribute Rule (AR) of \cite{gul2014random}. Both EBA and AR are RUM and generalize the Luce model. While each of these models has an intersection with the NSC, neither one nests nor is nested by NSC.

In Tversky's EBA, each alternative is a collection of \emph{aspects}. The decision maker randomly selects one of these aspects from the aspects available in the menu, via a Luce rule, and eliminates alternatives that do not have the selected aspect. The decision maker repeats this procedure until a single alternative remains. EBA is conceptually similar to an $N$-step nested logit, where $N$ is the total number of alternatives.  The Linear NSC is a special case of EBA, but EBA is disjoint from nested logit. 

In the AR of \cite{gul2014random}, each alternative has many \emph{attributes}. A decision maker randomly selects one attribute from the attributes available in the menu via a Luce rule. When the selected attribute is $\omega$, alternative $x$ will be chosen with a probability that is proportional to $\gamma^\omega_x$, where $\gamma^\omega_x$ is the intensity of $\omega$ in $x$. The AR is conceptually similar to cross-nested logit. In fact, one can show that the AR is a special case of a non-parametric version of cross-nested logit in which weights assigned to nests are menu-independent (i.e., $\gamma^\omega_x$ is menu-independent). Because of the menu independence of $\gamma^\omega_x$, the AR is more restrictive than unrestricted cross-nested logit.  The Linear NSC is a special case of the AR, but nested logit is not. 

Other recent papers dealing with the similarity effect are \cite{farolucereplicates} and \cite{li2016associationistic}. Both are special cases of NSC, but are generally distinct from nested logit. 

 \cite{farolucereplicates} introduces the \emph{Luce Model with Replicas} (LR), which is a special case of Linear NSC in which nest values and Luce utilities are constant: $v(A)=v_i$ for each $A\subseteq X_i$ and $u(a)=u(b)$ for any $a, b\in X_i$. In terms of behavior, Faro's model only allows for restrictive forms of the similarity effect in which similar alternatives are replicas. 
 
 \cite{li2016associationistic} present the \emph{Associationistic Luce Model} (AL), which is a special case of NSC with $v(A)=\sum_{a\in A} \gamma(a)$ for some function $\gamma$. Because of the additive structure of $v$, AL is significantly more restrictive than NSC. In fact, the Luce model is the only intersection between nested logit and AL. We note that the AL also allows for violations of regularity (e.g., the attraction effect). However, since $v$ is increasing, the AL cannot simultaneously allow for violations of regularity and the similarity effect (see Appendix \ref{monNSC}). In terms of axiomatic foundations, they also use the revealed similarity relation $\sim_p$, and impose transitivity of $\sim_p$ as one of their axioms.

 NSC has a large overlap with RUM, which goes back to \cite{Block1960rum}, \cite{falmagne1978representation}, and \cite{barbera1986falmagne}. For example, both random utility nested logit and Linear NSC are RUM. %The nested logit goes back to  \cite{Ben-Akiva1973nestedlogit} and  \cite{mcfadden1978modeling} and has been widely applied in many fields, as mentioned in the introduction. 
In addition to EBA, AR, and random utility nested logit, many special cases of RUM have been proposed, including: \cite{gul2006reu}, in which each preference has an expected utility representation; \cite{apesteguia2017scrum}, in which the collection of preferences satisfy the single-crossing property; and \cite{manzini2014stochastic}, in which randomness occurs due to stochastic consideration. Our characterization of random utility nested logit contributes to this area of the stochastic choice literature.

NSC has an interpretation as a sequential choice model, in which a nest is chosen and then an alternative.  \cite{manzini2012categorize} study a deterministic choice model in which a decision maker categorizes alternatives before choosing. The decision maker first selects the ``best'' category according to some ordering, then selects their most preferred alternative according to another.  Categories however do not need to form a partition, unlike in NSC.  \cite{ravid2018focus} introduce the following stochastic choice model that involves a sequence of binary comparisons: \[p(x, A)=\frac{\prod_{y\in A\setminus\{x\}}\pi(x, y)}{\sum_{z\in A} \prod_{t\in A\setminus\{z\}}\pi(z, t)}.\]
This model, which is a special case of \cite{marley1991}, is disjoint from nested logit but has an interesting connection to NSC. In particular, when $\pi(x, y)=\frac{1}{u(y)}$ and $\pi(x, z)=\frac{1}{w(z)}$ for any $x, y\in X_i$ and $z\in X_j$, we obtain an NSC with $v(A\cap X_i)=\big(\sum_{y\in A\cap X_i} u(y)\big)\frac{\prod_{y\in A\cap X_i} w(y)}{\prod_{y\in A\cap X_i} u(y)}$.

\bibliographystyle{ecta}
\bibliography{econref}

\appendix

\section{Proofs}

\subsection{Proof of \autoref{NSCthm}}

\noindent\textbf{Sufficiency.} We first prove the sufficiency part of Theorem 1 by the following nine steps. Suppose $p$ satisfies \nameref{ISA}. 

\medskip
\noindent\textbf{Step 1:} $\sim_p$ is transitive.\medskip

Take any $x, y, z\in X$ such that $x\sim_p y$ and $y\sim_p z$. We shall prove that $\frac{p(x, \{x, z\})}{p(z, \{x, z\})}=\frac{p(x, A)}{p(z, A)}$ for any $A$ with $x, z\in A$.

Take any $A$ with $x, z\in A$ and $y\not\in A$. Since $x\sim_p y$ and $y\sim_p z$, by \nameref{ISA} we have  $\frac{p(x, A\cup y)}{p(z, A\cup y)}=\frac{p(x, A)}{p(z, A)}$ and $\frac{p(x, \{x, y, z\})}{p(z, \{x, y, z\})}=\frac{p(x, \{x, z\})}{p(z, \{x, z\})}$. By the definition of $\sim_p$, $x\sim_p y$ implies $\frac{p(x, A\cup y)}{p(y, A\cup y)}=\frac{p(x, \{x, y, z\})}{p(y, \{x, y, z\})}$ and $y\sim_p z$ implies $\frac{p(y, A\cup y)}{p(z, A\cup y)}=\frac{p(y, \{x, y, z\})}{p(z, \{x, y, z\})}$. By combining all the previous equalities, 
\begin{eqnarray*}
\frac{p(x, A)}{p(z, A)}&=&\frac{p(x, A\cup y)}{p(z, A\cup y)}=\frac{p(x, A\cup y)}{p(y, A\cup y)}\cdot\frac{p(y, A\cup y)}{p(z, A\cup y)}\\
&=&\frac{p(x, \{x, y, z\})}{p(y, \{x, y, z\})}\cdot\frac{p(y, \{x, y, z\})}{p(z, \{x, y, z\})}=\frac{p(x, \{x, y, z\})}{p(z, \{x, y, z\})}=\frac{p(x, \{x, z\})}{p(z, \{x, z\})}.\end{eqnarray*}
Hence, $x\sim_p z$.

\medskip
\noindent\textbf{Step 2:} Let $X/\sim_p\equiv \{X_i\}^K_{i=1}$; that is, for any $x_i, x'_i\in X_i$ and $x_j\in X_j$, $x_i\sim_p x'_i$ and $x_i\not\sim_p x_j$. Since $\sim_p$ is reflexive, transitive, and symmetric, we have a well-defined partition of $X$.\medskip

\medskip
\noindent\textbf{Step 3:} The construction of $u$.\medskip

Notice that for each $i\le K$, IIA is satisfied at all subsets of $X_i$. Therefore, for each $i$, there is a utility function $u_i:X_i\to\mathbb{R}_{++}$ such that $p(a, A)=\frac{u_i(a)}{\sum_{b\in A}u_i(b)}$ for any $A\subseteq X_i$ and $a\in A$ (as in the characterization of the Luce model). Since $X_1, \ldots, X_k$ are disjoint, we also have $u:X\to\mathbb{R}_{++}$ such that for any $A\subseteq X_i$ and $a\in A$, $p(a, A)=\frac{u(a)}{\sum_{b\in A}u(b)}$.

\medskip
\noindent\textbf{Step 4:} For any $A\in\mathscr{A}$ and $a\in A\cap X_i$, 
\[p(a, A)=\frac{u(a)}{\sum_{x\in A\cap X_i}u(x)}\,p(A\cap X_i, A).\]

Take any $A$ and $a\in A\cap X_i$. By the definitions of $\sim_p$, $u$, and $\{X_j\}^K_{j=1}$, we have $\frac{p(a', A)}{p(a, A)}=\frac{p(a', \{a', a\})}{p(a, \{a', a\})}=\frac{u(a')}{u(a)}$ for any $a'\in A\cap X_i$. Then $\frac{p(A\cap X_i, A)}{p(a, A)}=\frac{\sum_{a'\in A\cap X_i}u(a')}{u(a)}$. Hence, $p(a, A)=\frac{u(a)}{\sum_{a'\in A\cap X_i}u(a')}\,p(A\cap X_i, A)$.\medskip

\medskip
\noindent\textbf{Step 5.} Take alternatives $a, b, x$ such that $a\in X_i, b\in X_j$, and $x\in X_k$. By \nameref{ISA}, for any $A\in\mathscr{A}$ with $a, b, x\in A$, we have $\frac{p(a, A)}{p(b, A)}=\frac{p(a, A\setminus \{x\})}{p(b, A\setminus \{x\})}$ since $a\not\sim_p x$ and $b\not\sim_p x$. Equivalently, 
\[\frac{\frac{u(a)}{\sum_{y\in A\cap X_i}u(y)}}{\frac{u(b)}{\sum_{z\in A\cap X_j}u(z)}}\cdot\frac{p(A\cap X_i, A)}{p(A\cap X_j, A)}=
\frac{\frac{u(a)}{\sum_{y\in A\cap X_i}u(y)}}{\frac{u(b)}{\sum_{z\in A\cap X_j}u(z)}}\cdot\frac{p(A\cap X_i, A\setminus \{x\})}{p(A\cap X_j, A\setminus \{x\})}.\]
Therefore, for any $A\in\mathscr{A}$ and $a, b, x\in A$ such that $a\in X_i, b\in X_j$, and $x\in X_k$, 
\begin{equation}\label{IICeqn}\frac{p(A\cap X_i, A)}{p(A\cap X_j, A)}=\frac{p(A\cap X_i, A\setminus \{x\})}{p(A\cap X_j, A\setminus \{x\})}.\end{equation}

\medskip
\noindent\textbf{Step 6.} We will construct the nest utility function $v:\bigcup^K_{i=1} 2^{X_i}\to\mathds{R}_{+}$ for subsets of $X_2, \ldots, X_k$ in the following way.\medskip

First, let us take $\alpha_i\subseteq X_i$ where $i\ge 2$. Let
\[v(\alpha_i)\equiv\frac{p(\alpha_i, \alpha_i \cup X_1)}{p(X_1, \alpha_i \cup X_1)}.\]

\noindent\textbf{Fact 1.} For any $\alpha_i\subseteq X_i, \alpha_j\subseteq X_j$ with $i, j\ge 2$,

\[\frac{p(\alpha_i, \alpha_i \cup \alpha_j)}{p(\alpha_j, \alpha_i \cup \alpha_j)}=\frac{v(\alpha_i)}{v(\alpha_j)}.\]

\noindent\textbf{Proof of Fact 1.} Notice that from Equation (\ref{IICeqn}) we can obtain the following by repeatedly eliminating $x\in X_1$ from $\alpha_i\cup \alpha_j\cup X_1$:
\[\frac{p(\alpha_i, \alpha_i \cup\alpha_j\cup X_1)}{p(\alpha_j, \alpha_i\cup \alpha_j\cup X_1)}=\frac{p(\alpha_i, \alpha_i\cup \alpha_j)}{p(\alpha_j, \alpha_i \cup\alpha_j)}.\]
Similarly, from Equation (\ref{IICeqn}) we obtain
\[\frac{p(\alpha_i, \alpha_i\cup \alpha_j \cup X_1)}{p(X_1, \alpha_i \cup\alpha_j \cup X_1)}=\frac{p(\alpha_i, \alpha_i \cup X_1)}{p(X_1, \alpha_i \cup X_1)}=v(\alpha_i)\]
and \[\frac{p(\alpha_j, \alpha_i\cup \alpha_j\cup X_1)}{p(X_1, \alpha_i\cup \alpha_j\cup X_1)}=\frac{p(\alpha_j, \alpha_j\cup X_1)}{p(X_1, \alpha_j\cup X_1)}=v(\alpha_j).\]
Combining the above three equalities, we obtain
\[\frac{v(\alpha_i)}{v(\alpha_j)}=\frac{\frac{p(\alpha_i, \alpha_i\cup \alpha_j\cup X_1)}{p(X_1, \alpha_i\cup \alpha_j\cup X_1)}}{\frac{p(\alpha_j, \alpha_i\cup \alpha_j\cup X_1)}{p(X_1, \alpha_i\cup \alpha_j\cup X_1)}}=\frac{p(\alpha_i, \alpha_i\cup \alpha_j\cup X_1)}{p(\alpha_j, \alpha_i\cup \alpha_j\cup X_1)}=\frac{p(\alpha_i, \alpha_i\cup \alpha_j)}{p(\alpha_j, \alpha_i\cup \alpha_j)}.\]

\noindent\textbf{Fact 2.} For any $\alpha_i\subseteq X_i, \alpha_j\subseteq X_j$ with $i, j\ge 2$ and $A\subseteq \cup_{s\neq i, j} X_s$,

\[\frac{p(\alpha_i, \alpha_i\cup \alpha_j\cup A)}{p(\alpha_j, \alpha_i\cup \alpha_j\cup A) }=\frac{v(\alpha_i)}{v(\alpha_j)}.\]

\noindent\textbf{Proof of Fact 2.} By Equation (\ref{IICeqn}), we obtain the following equality by repeatedly eliminating $x\in A$ from $\alpha_i\cup \alpha_j\cup A$: 
\[\frac{p(\alpha_i, \alpha_i\cup \alpha_j\cup A)}{p(\alpha_j, \alpha_i\cup \alpha_j\cup A) }=\frac{p(\alpha_i, \alpha_i\cup \alpha_j)}{p(\alpha_j, \alpha_i\cup \alpha_j)}=\frac{v(\alpha_i)}{v(\alpha_j)}.\]

\medskip
\noindent\textbf{Step 7.} We will construct the nest utility function $v:\bigcup^K_{i=1} 2^{X_i}\to\mathds{R}_{+}$ for subsets of $X_1$ in the following way.\medskip

First, let us take $\alpha_1\subseteq X_1$. Let
\[v(\alpha_1)\equiv\frac{p(\alpha_1, \alpha_1\cup X_2)}{p(X_2, \alpha_1\cup X_2)}\cdot \frac{p(X_2, X_1\cup X_2)}{p(X_1, X_1\cup X_2)}.\]

\noindent\textbf{Fact 3.} For any $\alpha_2\subseteq X_2$,

\[\frac{p(\alpha_1, \alpha_1\cup \alpha_2)}{p(\alpha_2, \alpha_1\cup \alpha_2)}=\frac{v(\alpha_1)}{v(\alpha_2)}.\]

\noindent\textbf{Proof of Fact 3.} Since $v(\alpha_2)=\frac{p(\alpha_2, X_1\cup \alpha_2)}{p(X_1, X_1\cup \alpha_2)}$, we shall prove that 
\[\frac{p(\alpha_1, \alpha_1\cup X_2)}{p(X_2, \alpha_1\cup X_2)}\cdot \frac{p(X_2, X_1\cup X_2)}{p(X_1, X_1\cup X_2)}=\frac{p(\alpha_1, \alpha_1\cup \alpha_2)}{p(\alpha_2, \alpha_1\cup \alpha_2)}\cdot \frac{p(\alpha_2, X_1\cup \alpha_2)}{p(X_1, X_1\cup \alpha_2)}.\]

Notice that from Equation (\ref{IICeqn}) we can obtain the following equalities by repeatedly eliminating $x\in X_3$ from $\alpha_1\cup X_2\cup X_3$ and $\alpha_1\cup \alpha_2\cup X_3$,
\[\frac{p(\alpha_1, \alpha_1\cup X_2)}{p(X_2, \alpha_1\cup X_2)}\cdot \frac{p(X_2, X_1\cup X_2)}{p(X_1, X_1\cup X_2)}=\frac{p(\alpha_1, \alpha_1\cup X_2\cup X_3)}{p(X_2, \alpha_1\cup X_2\cup X_3)}\cdot \frac{p(X_2, X_1\cup X_2)}{p(X_1, X_1\cup X_2)}\]
and 
\[\frac{p(\alpha_1, \alpha_1\cup \alpha_2)}{p(\alpha_2, \alpha_1\cup \alpha_2)}\cdot \frac{p(\alpha_2, X_1\cup \alpha_2)}{p(X_1, X_1\cup \alpha_2)}=\frac{p(\alpha_1, \alpha_1\cup \alpha_2\cup X_3)}{p(\alpha_2, \alpha_1\cup \alpha_2\cup X_3)}\cdot \frac{p(\alpha_2, X_1\cup \alpha_2)}{p(X_1, X_1\cup \alpha_2)}.\]
Therefore, we shall prove that 
\[\frac{p(\alpha_1, \alpha_1\cup X_2\cup X_3)}{p(X_2, \alpha_1\cup X_2\cup X_3)}\cdot \frac{p(X_2, X_1\cup X_2)}{p(X_1, X_1\cup X_2)}=\frac{p(\alpha_1, \alpha_1\cup \alpha_2\cup X_3)}{p(\alpha_2, \alpha_1\cup \alpha_2\cup X_3)}\cdot \frac{p(\alpha_2, X_1\cup \alpha_2)}{p(X_1, X_1\cup \alpha_2)}.\]

Moreover,
\begin{eqnarray*}
\!\!\frac{p(\alpha_1, \alpha_1\cup X_2\cup X_3)}{p(X_2, \alpha_1\cup X_2\cup X_3)}\cdot \frac{p(X_2, X_1\cup X_2)}{p(X_1, X_1\cup X_2 )}\!\!&\!=\!&\!\!\frac{p(\alpha_1, \alpha_1\cup X_2\cup X_3)}{p(X_3, \alpha_1\cup X_2\cup X_3)}\cdot \frac{p(X_3, \alpha_1\cup X_2\cup X_3)}{p(X_2, \alpha_1\cup X_2\cup X_3)}\cdot \frac{p(X_2, X_1\cup X_2)}{p(X_1, X_1\cup X_2)}\\
\!\!&\!=\!&\!\!\frac{p(\alpha_1, \alpha_1\cup X_2\cup X_3)}{p(X_3, \alpha_1\cup X_2\cup X_3)}\cdot \frac{v(X_3)}{v(X_2)}\cdot \frac{p(X_2, X_1\cup X_2)}{p(X_1, X_1\cup X_2)}\text{, by Fact 2,}\\
\!\!&\!=\!&\!\!\frac{p(\alpha_1, \alpha_1\cup X_2\cup X_3)}{p(X_3, \alpha_1\cup X_2\cup X_3)}\cdot \frac{v(X_3)}{v(X_2)}\cdot v(X_2)\text{, by the definition of }v,\end{eqnarray*}and 
\begin{eqnarray*}
\frac{p(\alpha_1, \alpha_1\cup \alpha_2\cup X_3)}{p(\alpha_2, \alpha_1\cup \alpha_2\cup X_3)}\cdot \frac{p(\alpha_2, X_1\cup \alpha_2 )}{p(X_1, X_1\cup \alpha_2)}\!\!&=&\!\!\frac{p(\alpha_1, \alpha_1\cup \alpha_2\cup X_3)}{p(X_3, \alpha_1\cup \alpha_2\cup X_3)}\cdot\frac{p(X_3, \alpha_1\cup \alpha_2\cup X_3)}{p(\alpha_2, \alpha_1\cup \alpha_2\cup X_3)}\cdot \frac{p(\alpha_2, X_1\cup \alpha_2)}{p(X_1, X_1\cup \alpha_2)}\\
\!\!&=&\!\!\frac{p(\alpha_1, \alpha_1\cup \alpha_2\cup X_3)}{p(X_3, \alpha_1\cup \alpha_2\cup X_3)}\cdot\frac{v(X_3)}{v(\alpha_2)}\cdot \frac{p(\alpha_2, X_1\cup \alpha_2)}{p(X_1, X_1\cup\alpha_2)}\text{, by Fact 2,}\\
\!\!&=&\!\!\frac{p(\alpha_1, \alpha_1\cup \alpha_2\cup X_3)}{p(X_3, \alpha_1\cup\alpha_2\cup X_3)}\cdot\frac{v(X_3)}{v(\alpha_2)}\cdot v(\alpha_2)\text{, by the definition of $v$.}\end{eqnarray*}Finally, we shall prove that 
\begin{equation}\frac{p(\alpha_1, \alpha_1\cup X_2\cup X_3)}{p(X_3, \alpha_1\cup X_2\cup X_3)}=\frac{p(\alpha_1, \alpha_1\cup \alpha_2\cup X_3)}{p(X_3, \alpha_1\cup \alpha_2\cup X_3)},\end{equation}
which immediately follows from Equation (\ref{IICeqn}) by repeatedly eliminating $x\in X_2\setminus \alpha_2$ from $\alpha_1\cup X_2\cup X_3$.  

\medskip
\noindent\textbf{Fact 4.} For any $\alpha_i\subseteq X_i$ with $i\ge 3$,

\[\frac{p(\alpha_1, \alpha_1\cup \alpha_i)}{p(\alpha_i, \alpha_1\cup \alpha_i)}=\frac{v(\alpha_1)}{v(\alpha_i)}.\]
\noindent\textbf{Proof of Fact 4.} By Equation (\ref{IICeqn}) and Facts 2-3,
\begin{eqnarray*}
\frac{p(\alpha_1, \alpha_1\cup \alpha_i)}{p(\alpha_i, \alpha_1\cup \alpha_i)}&=&
\frac{p(\alpha_1, \alpha_1\cup \alpha_i\cup X_2)}{p(\alpha_i, \alpha_1\cup \alpha_i\cup X_2)}=\frac{p(\alpha_1, \alpha_1\cup \alpha_i\cup X_2)}{p(X_2, \alpha_1\cup \alpha_i\cup X_2)}\cdot \frac{p(X_2, \alpha_1\cup \alpha_i\cup X_2)}{p(\alpha_i, \alpha_1\cup \alpha_i\cup X_2)}\\
&=&\frac{p(\alpha_1, \alpha_1\cup X_2)}{p(X_2, \alpha_1\cup X_2)}\cdot \frac{v(X_2)}{v(\alpha_i)}=\frac{v(\alpha_1)}{v(X_2)}\cdot \frac{v(X_2)}{v(\alpha_i)}=\frac{v(\alpha_1)}{v(\alpha_i)}.
\end{eqnarray*}

\medskip
\noindent\textbf{Fact 5.} For any $\alpha_i\subseteq X_i$ with $i\ge 2$ and $A\subseteq \cup_{s\neq i, 1} X_s$,

\[\frac{p(\alpha_1, \alpha_1\cup \alpha_i\cup A)}{p(\alpha_i, \alpha_1\cup \alpha_i\cup A)}=\frac{v(\alpha_1)}{v(\alpha_i)}.\]

\noindent\textbf{Proof of Fact 5.} From Equation (\ref{IICeqn}), we have \[\frac{p(\alpha_1, \alpha_1\cup \alpha_i\cup A)}{p(\alpha_i, \alpha_1\cup \alpha_i\cup A)}=\frac{p(\alpha_1, \alpha_1\cup \alpha_i)}{p(\alpha_i, \alpha_1\cup \alpha_i)}=\frac{v(\alpha_1)}{v(\alpha_i)}.\]

\noindent\textbf{Step 8.} By Facts 2, 4, and 5, for any $A\in\mathscr{A}$ and $i, j\le K$,
\[\frac{p(A\cap X_i, A)}{p(A\cap X_j, A)}=\frac{v(A\cap X_i)}{v(A\cap X_j)}.\]
Since $\sum^K_{i=1}p(A\cap X_i, A)=1$, we have $p(A\cap X_i, A)=\frac{v(A\cap X_i)}{\sum^K_{j=1} v(A\cap X_j)}$. Since $\frac{p(a, A)}{p(b, A)}=\frac{u(a)}{u(b)}$ for any $a, b\in X_i$, we have 
\[p(a_i, A)=\frac{u(a_i)}{\sum_{x\in A\cap X_i}u(x)}\cdot\frac{v(A\cap X_i)}{\sum^K_{j=1} v(A\cap X_j)}\text{ for each }a_i\in A\cap X_i.\]

\noindent\textbf{Step 9.} An NSC $p$ is nondegenerate.\medskip

By way of contradiction, suppose there are $i, j\le K$ such that for some $a\in X_i$ and $b\in X_j$,  
\[\frac{\sum_{x\in A_i}u(x)}{v(A_i)}=\frac{u(a)}{v(a)}\text{ and }\frac{\sum_{y\in A_j}u(y)}{v(A_j)}=\frac{u(b)}{v(b)}\text{ for any }A_i\subseteq X_i\text{ and }A_j\subseteq X_j.\]
In other words, 
\[\frac{\sum_{x\in A_i}u(x)}{v(A_i)}\Big{/}\frac{u(a)}{v(a)}= \frac{\sum_{y\in A_j}u(y)}{v(A_j)}\Big{/}\frac{u(b)}{v(b)}\text{ for any }A_i\subseteq X_i\text{ and }A_j\subseteq X_j.\]
Then by NSC representation, we have $\frac{p(a, \{a, b\})}{p(b, \{a, b\})}=\frac{p(a, A)}{p(b, A)}$ for any $A\in\mathscr{A}$; i.e., $a\sim_p b$, which contradicts the construction of $\{X_k\}^K_{k=1}$.\bigskip 

\noindent\textbf{Necessity.} Suppose $p$ is a nondegenerate NSC with $(v, u, \{X_i\}^K_{i=1})$.\medskip 

\noindent\textbf{Step 1.} For any $a, b\in X$, $a\sim_p b$ if and only if either $a, b\in X_i$ for some $i$.\medskip 

Take any $a, b\in X$. We consider two cases.\medskip

\noindent\textbf{Case 1}. Suppose $a, b\in X_i$. \medskip

In this case, by NSC representation, $\frac{p(a, A)}{p(b, A)}=\frac{u(a)}{u(b)}=\frac{p(a, \{a, b\})}{p(b, \{a, b\})}$ for any $A\in\mathscr{A}$. Therefore, $a\sim_p b$.\medskip

\noindent\textbf{Case 2}. Suppose $a\in X_i$ and $b\in X_j$ with $i\neq j$.\medskip

We shall prove that $a\not\sim_p b$. By nondegeneracy of $p$, either 
\[\frac{\sum_{x\in A_i}u(x)}{v(A_i)}\neq \frac{u(a)}{v(a)}\text{ for some }A_i \subseteq X_i \text{ with }a\in A_i\]
or
\[\frac{\sum_{y\in A_j}u(y)}{v(A_j)}\neq \frac{u(b)}{v(b)}\text{ for some }A_j\subseteq X_j\text{ wtih }b\in A_j.\]
Without loss of generality, suppose the former is true. Then we have
\[\frac{p(a, \{a, b\})}{p(b, \{a, b\})}=\frac{v(a)}{v(b)}\neq \frac{p(a, A_i\cup b)}{p(b, A_i\cup b)}=\frac{u(a)}{\sum_{x\in A_i}u(x)}\frac{v(A_i)}{v(b)}.\]
Therefore, $a\not\sim_p b$.

\medskip
\noindent\textbf{Step 2.} The first part of \nameref{ISA} is satisfied.\medskip

Take any $A\in\mathscr{A}$, $a, b\in A$, and $x\not\in A$ such that $a\sim_p x$ and $b\sim_p x$. By Step 1, we have $a, b, x\in X_i$ for some $i$. Therefore, $a\sim_p b$ implies $\frac{p(a, A)}{p(b, A)}=\frac{p(a, \{a, b\})}{p(b, \{a, b\})}=\frac{p(a, A\cup x)}{p(b, A\cup x)}$.

\medskip
\noindent\textbf{Step 3.} The second part of \nameref{ISA} is satisfied.\medskip

Take any $A\in\mathscr{A}$, $a, b\in A$, and $x\not\in A$ such that $a\not\sim_p x$ and $b\not\sim_p x$. By Step 1, $a\not\sim_p x$ and $b\not\sim_p x$ imply $x\in X_i$ and $a, b\not\in X_i$ for some $i$.

\medskip
\noindent\textbf{Case 1.} $a, b\in X_j$ for some $j$.\medskip

Since $a\sim_p b$, 
\[\frac{p(a, A)}{p(b, A)}=\frac{p(a, \{a, b\})}{p(b, \{a, b\})}=\frac{p(a, A\cup x)}{p(b, A\cup x)}.\]    

\noindent\textbf{Case 2.} $a\in X_j$ and $b\in X_k$ for some $j, k$ with $j\neq k$.\medskip    

In this case, we have

\[\frac{p(a, A)}{p(b, A)}=\frac{\frac{u(a)}{\sum_{y\in A_j}u(y)}\,v(A_j)}{\frac{u(b)}{\sum_{z\in A_k}u(z)}\,v(A_k)}=\frac{p(a, A\cup x)}{p(b, A\cup x)}\]
since $A_j=(A\cup x)\cap X_j=A\cap X_j$ and $A_k=(A\cup x)\cap X_k=A\cap X_k$.

\subsection{Proof of \autoref{NSCuniqueprop}}

Since the uniqueness of Luce utilities and nest utilities follows standard arguments, we only prove that the nest structure is unique. By way of contradiction, suppose $p$ is a nondegenerate NSC with respect to both of $(v, u, \{X_i\}^K_{i=1})$ and $(v', u', \{X'_i\}^{K'}_{i=1})$, and $\{X'_i\}^{K'}_{i=1}$ is not a permutation of $\{X_i\}^{K}_{i=1}$. Without loss of generality, suppose $K'\ge K$. Then there are $x_i, x'_i\in X_i$ such that $x_i\in X'_j$ and $x'_i\in X'_k$. Since $\frac{p(x_i, \{x_i, x'_i\})}{p(x'_i, \{x_i, x'_i\})}=\frac{p(x_i, A)}{p(x'_i, A)}$ for any $A$, by NSC representation with $\{v', u', \{X'_i\}^{K'}_{i=1}\}$, we have
\[\frac{v'(x_i)}{v'(x'_i)}=\frac{\frac{u'(x_i)}{\sum_{x_j\in A_j} u'(x_j)} v'(A_j)}{\frac{u'(x'_i)}{\sum_{x_k\in A_k} u'(x_k)} v'(A_k)}\text{ for any }A_j\subseteq X'_j\text{ and }A_k\subseteq X'_k.\]

Let us first set $A_j=\{x_i\}$. Then we have $\frac{u'(x'_i)}{v'(x'_i)}=\frac{\sum_{x_k\in A_k} u'(x_k)}{v'(A_k)}\text{ for any }A_k\subseteq X'_k$.
Similarly, by setting $A_k=\{x'_i\}$, we also obtain $\frac{u'(x_i)}{v'(x_i)}=\frac{\sum_{x_j\in A_j} u'(x_j)}{v'(A_j)}\text{ for any }A_j\subseteq X'_j$. Therefore, we obtain a contradiction since the above two equalities contradict the assumption that NSC $p$ with $(v', u', \{X'_i\}^{K'}_{i=1})$ is nondegenerate.

\subsection{Proof of \autoref{NL}}

Since the necessity part is straightforward, we only prove the sufficiency part. Suppose $p$ is an NSC with $(v, u, \{X_i\}^K_{i=1})$ and that it satisfies \nameref{LRI}. If $K=1$, then we immediately have a Luce model. Suppose now $K\ge 2$. For each $i\le K$, fix some $a^*_i \in X_i$. Then define \[\eta_i\equiv\frac{\log\big(v(X_i)/v(\{a^*_i\})\big)}{\log\big(\sum_{a_i \in X_i}u(a_i)/u(a^*_i)\big)}.\]

For any $A \subset X_i$ and any $x \in X\setminus X_i$, by Log Ratio Invariance, we have

\[\frac{\log\Big(\frac{p(X_i,\, X_i\,\cup\, x)}{p(x,\, X_i\,\cup\, x)}\big{/}\frac{p(a^*_i,\, \{a^*_i,\, x\})}{p(x,\, \{a^*_i,\, x\})}\Big)}{\log\Big(\frac{p(X_i, \,X_i\,\cup\, a^*_i)}{p(a^*_i, \,X\,\cup\, a^*_i)}\Big)}=\frac{\log\Big(\frac{p(A,\, A\,\cup\, x)}{p(x,\, A\,\cup\, x)}\big{/}\frac{p(a^*_i,\, \{a^*_i,\, x\})}{p(x,\, \{a^*_i,\, x\})}\Big)}{\log\Big(\frac{p(A, \,A\,\cup\, a^*_i)}{p(a^*_i, \,A\,\cup\, a^*_i)}\Big)}.\] 
Since 
\[\frac{\log\Big(\frac{p(X_i,\, X_i\,\cup\, x)}{p(x,\, X_i\,\cup\, x)}\big{/}\frac{p(a^*_i,\, \{a^*_i,\, x\})}{p(x,\, \{a^*_i,\, x\})}\Big)}{\log\Big(\frac{p(X_i, \,X_i\,\cup\, a^*_i)}{p(a^*_i, \,X\,\cup\, a^*_i)}\Big)}=\frac{\log\Big(\frac{v(X_i)}{v(\{a^*_i\})}\Big)}{\log\Big(\frac{\sum_{b \in X_i}u(b)}{u(a^*_i)}\Big)}= \eta_i,\]
we have 
\[\eta_i=\frac{\log\Big(\frac{p(A,\, A\,\cup\, x)}{p(x,\, A\,\cup\, x)}\big{/}\frac{p(a^*_i,\, \{a^*_i,\, x\})}{p(x,\, \{a^*_i,\, x\})}\Big)}{\log\Big(\frac{p(A, \,A\,\cup\, a^*_i)}{p(a^*_i, \,A\,\cup\, a^*_i)}\Big)}=\frac{\log\Big(\frac{v(A)}{v(\{a^*_i\})}\Big)}{\log\Big(\frac{\sum_{b \in A}u(b)}{u(a^*_i)}\Big)};\]
equivalently, $\frac{v(A)}{v(\{a^*_i\})}=(\frac{\sum_{b \in A}u(b)}{u(a^*_i)})^{\eta_i}$. Let $\delta_i=\frac{v(\{a^*_i\})}{(u(a^*_i))^{\eta_i}}$. Then $v(A)=\delta_i\,(\sum_{b \in A}u(b))^{\eta_i}$. Therefore, $p$ is the nested logit with $(\eta_1, \ldots, \eta_K, u', \{X_i\}^K_{i=1})$ such that $u'(x)=\delta^{\frac{1}{\eta_i}}_i\,u(x)$ when $x\in X_i$.

\subsection{Proof of \autoref{proposition2}}

Suppose $p$ is a nested logit with $(\eta_1, \ldots, \eta_K, u, \{X_i\}^K_{i=1})$. Take any $A, A'\in \mathscr{A}$ such that IIA is satisfied at $A\cup A'$. We shall prove that Relative Likelihood Independence is satisfied. Since $p$ is a nondegenerate NSC, by Step 1 of the necessity part proof of Theorem 1, $x\sim_p y$ if and only if $x, y\in X_i$ for some $i\le K$. Take any $A, B, A', B'\in\mathscr{A}$ and $x\in X$ such that $a\sim_p a'$ for any $a, a'\in A\cup B\cup A'\cup B'$. By the previous argument, $A\cup B\cup A'\cup B'\subseteq X_i$. By the nested logit representation,

\[\frac{p(A, A\cup B)}{p(B, A\cup B)}=\frac{\sum_{a\in A} u(a)}{\sum_{b\in B} u(b)}\ge \frac{p(A', A'\cup B')}{p(B', A'\cup B')}=\frac{\sum_{a\in A'} u(a)}{\sum_{b\in B'} u(b)}.\]
When $x\in X_i$,
\[\frac{p(A, A\cup x)}{p(x, A\cup x)}\big{/}\frac{p(B, B\cup x)}{p(x, B\cup x)}=\frac{\sum_{a\in A} u(a)}{\sum_{b\in B} u(b)}
\ge \frac{p(A', A'\cup x)}{p(x, A'\cup x)}\big{/}\frac{p(B', B'\cup x)}{p(x, B'\cup x)}=\frac{\sum_{a\in A'} u(a)}{\sum_{b\in B'} u(b)}.\]
When $x\in X_j$ and $i\neq j$,
\[\frac{p(A, A\cup x)}{p(x, A\cup x)}\big{/}\frac{p(B, B\cup x)}{p(x, B\cup x)}=\frac{(\sum_{a\in A} u(a))^{\eta_i}}{(\sum_{b\in B} u(b))^{\eta_i}}
\ge \frac{p(A', A'\cup x)}{p(x, A'\cup x)}\big{/}\frac{p(B', B'\cup x)}{p(x, B'\cup x)}=\frac{(\sum_{a\in A'} u(a))^{\eta_i}}{(\sum_{b\in B'} u(b))^{\eta_i}}.\]
Hence, \[\frac{p(A, A\cup B)}{p(B, A\cup B)}\ge \frac{p(A', A'\cup B')}{p(B', A'\cup B')}\text{ implies }
\frac{p(A, A\cup x)}{p(x, A\cup x)}\big{/}\frac{p(B, B\cup x)}{p(x, B\cup x)}\ge \frac{p(A', A'\cup x)}{p(x, A'\cup x)}\big{/}\frac{p(B', B'\cup x)}{p(x, B'\cup x)}.\]

Suppose $p$ is the nondegenerate NSC with $(v, u, \{X_i\}^K_{i=1})$ and satisfies Relative Likelihood Independence. If $K=1$, we trivially obtain the desired result. Suppose $K\ge 2$. Take any $i, j\le K$ with $i\neq j$. Take any $A, B, A'\subseteq X_i$ and $x\in X_j$. By Relative Likelihood Independence,
\[\frac{p(A, A\cup B)}{p(B, A\cup B)}\ge 1\text{ implies }\frac{p(A, A\cup x)}{p(x, A\cup x)}\big{/}\frac{p(B, B\cup x)}{p(x, B\cup x)}\ge 1;\]
equivalently, \[p(A, A\cup B)\ge p(B, A\cup B)\text{ implies }\frac{p(A, A\cup x)}{p(x, A\cup x)}\ge \frac{p(B, B\cup x)}{p(x, B\cup x)}.\]
By the NSC representation, $\sum_{a\in A}u(a)\ge \sum_{b\in B}u(b)$ implies \[\frac{p(A, A\cup x)}{p(x, A\cup x)}=\frac{v(A)}{v(x)}\ge \frac{p(B, B\cup x)}{p(x, B\cup x)}=\frac{v(B)}{v(x)}.\]
Then $\sum_{a\in A}u(a)\ge \sum_{b\in B}u(b)$ implies $v(A)\ge v(B)$. Therefore, there is an increasing function $f_i:\mathds{R}_{++}\to\mathds{R}_{++}$ such that $v(A)=f_i\big(\sum_{a\in A}u(a)\big)$ for any $A\subseteq X_i$.

\subsection{Proof of \autoref{nlthm}}

Suppose $p$ is the nondegenerate NSC with $(v, u, \{X_i\}^K_{i=1})$ and satisfies Relative Likelihood Independence and Richness. We shall show that $p$ is a nested logit. If $K=1$, we trivially obtain the desired result. Suppose $K\ge 2$. Let us fix $i\le K$. 

Take any $j\le K$ with $i\neq j$. Take any $A, B, A', B'\subseteq X_i$ and $x\in X_j$. By Relative Likelihood Independence and the NSC representation,
\[\frac{p(A, A\cup B)}{p(B, A\cup B)}=\frac{\sum_{a\in A} u(a)}{\sum_{b\in B} u(b)}\ge \frac{p(A', A'\cup B')}{p(B', A'\cup B')}=\frac{\sum_{a\in A'} u(a)}{\sum_{b\in B'} u(b)}\]
\[\text{ implies }\frac{p(A, A\cup x)}{p(x, A\cup x)}\big{/}\frac{p(B, B\cup x)}{p(x, B\cup x)}=\frac{v(A)}{v(B)}\ge \frac{p(A', A'\cup x)}{p(x, A'\cup x)}\big{/}\frac{p(B', B'\cup x)}{p(x, B'\cup x)}=\frac{v(A')}{v(B')}.\]
Equivalently, \[\frac{\sum_{a\in A} u(a)}{\sum_{b\in B} u(b)}\ge \frac{\sum_{a\in A'} u(a)}{\sum_{b\in B'} u(b)}\text{ implies }\frac{v(A)}{v(B)}\ge \frac{v(A')}{v(B')}.\]
When $A'=B'$, $\sum_{a\in A}u(a)\ge \sum_{b\in B}u(b)$ implies $v(A)\ge v(B)$. Therefore, there is an increasing function $f_i:\mathds{R}_{++}\to\mathds{R}_{++}$ such that $v(A)=f_i\big(\sum_{a\in A}u(a)\big)$ for any $A\subseteq X_i$. Let $R_i=\{x\in\mathbb{R}_{++}| x=\sum_{a\in A} u(a) \text{ for some }A\in\mathscr{A}\}$. Then we have for any $v, v', w, w'\in R_i$, 
\[\frac{v}{v'}\ge \frac{w}{w'}\text{ implies } \frac{f_i(v)}{f_i(v')}\ge \frac{f_i(w)}{f_i(w')}.\]

Take any $a\in X_i$. By Richness, for any $\rho\in (0, 1)$, there is $b\in X_i$ such that $u(b)=\frac{1-\rho}{\rho}\, u(a)$. Therefore, for any $\alpha>0$, there is $b\in X_i$ such that $u(b)=\alpha$. Hence, $u(X_i)=\mathbb{R}_{++}$.

Take any $\alpha, \beta\in \mathbb{R}_{++}$. Since $u(X_i)=\mathbb{R}_{++}$, there are alternatives $a, b, a', b'\in X_i$ such that $u(a)=\alpha\beta, u(b)=\beta, u(a')=\alpha$, and $u(b')=1$. By the above implication of Relative Likelihood Consistency, we have
\[\frac{f_i(\alpha\,\beta)}{f_i(\beta)}=\frac{f_i(\alpha)}{f_i(1)}\text{ for any }\alpha, \beta>0.\]

Let $g(t)\equiv\frac{f_i(t)}{f_i(1)}$. Then $g(1)=1$ and $g(\alpha\,\beta)=g(\alpha)\,g(\beta)$ for any $\alpha, \beta>0$. Finally, we can prove that $g$ is a power function. Since $g>0$, let $h(t)=\log(g(\exp(t)))$ for any $t\in\mathds{R}$. Then for any $t, t'\in\mathds{R}$, we have $h(t+t')=\log(g(\exp(t+t')))=\log(g(\exp(t)\,\exp(t')))=\log(g(\exp(t))\,g(\exp(t')))=\log(g(\exp(t)))+\log(g(\exp(t')))=h(t)+h(t')$. We have obtained a typical Cauchy functional equation for $h$. Hence, there is $\eta_i\ge 0$ such that $h(t)=\eta_i\,t$. In other words, $g(a)=a^{\eta_i}$. Therefore,
\[p(a, A)=\frac{f_i(1)\,\big(\sum_{x\in A\cap X_i}u(x)\big)^{\eta_i}}{\sum_{j: A\cap X_j\neq\emptyset}f_j(1)\,\big(\sum_{y\in A\cap X_j}u(y)\big)^{\eta_j}}\frac{u(a)}{\sum_{b\in A\cap X_i} u(b)}.\]
\[=\frac{\big(\sum_{x\in A\cap X_i}\bar{u}(x)\big)^{\eta_i}}{\sum_{j: A\cap X_j\neq\emptyset}\,\big(\sum_{y\in A\cap X_j}\bar{u}(y)\big)^{\eta_j}}\frac{\bar{u}(a)}{\sum_{b\in A\cap X_i} \bar{u}(b)},\]
where $\bar{u}(x)=(f_i(1))^\frac{1}{\eta_i}\,u(x)$ for each $x\in X_i$. That is, $p$ is a nested logit.

\subsection{Proof of \autoref{runlprop}}

Suppose $p$ is the nested logit with $(\eta_1, \ldots, \eta_K, u, \{X_i\}^K_{i=1})$ and satisfies Regularity and Richness. If $K=1$, we obtain the desired result since the Luce model is a random utility nested logit. Suppose $K\ge 2$. Let us fix $i\le K$. We shall prove that $\eta_i\le 1$. Take any $a, a'\in X_i$ and $b'\in X_j$ with $i\neq j$. By Richness, for any $\rho\in (0, 1)$ there is $b\in X_j$ such that $u(b)=\frac{1-\rho}{\rho}\,u(b')$. By Regularity, $p(a, \{a, b\})\le p(a, \{a, b, a'\})$; i.e.,  

\[p(a, \{a, b\})=\frac{(u(a))^{\eta_i}}{(u(a))^{\eta_i}+(u(b))^{\eta_j}}\ge p(a, \{a, b, a'\})=\frac{u(a)}{u(a)+u(a')}\cdot\frac{(u(a)+u(a'))^{\eta_i}}{(u(a)+u(a'))^{\eta_i}+(u(b))^{\eta_j}}.\]
After simplifying the above inequality, we obtain 
\[\frac{(u(a)+u(a'))^{\eta_i}+(u(b))^{\eta_j}}{(u(a))^{\eta_i}+(u(b))^{\eta_j}}=1+\frac{(u(a)+u(a'))^{\eta_i}-(u(a))^{\eta_i}}{(u(a))^{\eta_i}+(u(b))^{\eta_j}}\ge \frac{(u(a)+u(a'))^{\eta_i-1}}{(u(a))^{\eta_i-1}};\]
equivalently, 
\[\frac{(u(a))^{\eta_i-1}}{(u(a))^{\eta_i}+(u(b))^{\eta_j}}\ge \frac{(u(a)+u(a'))^{\eta_i-1}-(u(a))^{\eta_i-1}}{(u(a)+u(a'))^{\eta_i}-(u(a))^{\eta_i}}.\]
Notice that when $\rho$ is close to $0$, we can obtain arbitrary large $u(b)$. Then the left-hand side of above inequality can be arbitrary close to zero. Therefore, the right-hand side must be negative. Hence, since $\eta_i>0$ implies $(u(a)+u(a'))^{\eta_i}-(u(a))^{\eta_i}>0$, we have $\eta_i\le 1$.

\subsection{Proof of \autoref{CNLthm}}

Before we proceed to the proof of \autoref{CNLthm}, it is useful to consider the following generalization of unrestricted cross-nested logit:
\begin{equation}\label{gcnl}
p(x, A)=\sum_{k: x\in A\cap X_k}\frac{u^k_x}{\sum_{y\in A\cap X_k}u^k_y}\cdot\frac{\Big(\sum_{y\in A\cap X_k}u^k_y\Big)^{\lambda}}{\sum_{l:A\cap X_l\neq\emptyset}\Big(\sum_{z\in A\cap X_l}u^l_z\Big)^{\lambda}}.\end{equation}

Note that the representation (\ref{gcnl}) reduces to unrestricted cross-nested logit by setting $u^k_x=\big(\alpha^k_x\, u(x)\big)^\frac{1}{\lambda}$. It turns out that the representation (\ref{gcnl}) is behaviorally equivalent to the unrestricted cross-nested logit.

\begin{lem}\label{cnl-md} Any stochastic choice $p$ that admits the representation (\ref{gcnl}) is an unrestricted cross-nested logit.
\end{lem}

\begin{proof}Suppose that $p$ admits the representation (\ref{gcnl}) with $\{X_k\}^K_{k=1}$, $\{u^k_x\}_{k\le K,\,x\in X}$, and $\lambda$. Let us define $u:X\to\mathds{R}_{++}$ and $\alpha^k_x$ as follows: for each $x\in X$ and $k\le K$,

\[u(x)\equiv\sum^K_{l=1}(u^l_x)^{\lambda}\text{ and }\alpha^k_x\equiv\frac{(u^k_x)^{\lambda}}{\sum^K_{l=1}(u^l_x)^{\lambda}}.\]
Then we have $u^k_x=(\alpha^k_x\,u(x))^{\frac{1}{\lambda}}$ and $\sum^K_{k=1}\alpha^k_x=1$. Hence, we obtain an unrestricted cross-nested logit.
\end{proof}

By \autoref{cnl-md}, we shall prove that any stochastic choice function $p$ admits the representation (\ref{gcnl}) with some $\{X_k\}^K_{k=1}, \{u^k_x\}_{k\le K,\,x\in X}$, and $\lambda$. We first set the collection of subsets $X_1, \ldots, X_K$ to be equivalent to $\mathscr{A}$. That is, for any $A\in\mathscr{A}$, there $i\le K$ such that $A=X_k$. Moreover, $X_i\neq X_j$ whenever $i\neq j$. Let us now write $u^A_x$ rather than $u^k_x$ when $A=X_k$. Then we shall find $U=\{u^A_x\}_{A\in\mathscr{A}, x\in A}\in\mathds{R}^N_{++}$ where $N\equiv \sum_{A\in\mathscr{A}}|A|$ and $\lambda>0$ such that
\begin{equation}
p(x, A)=\sum_{B: x\in B}\frac{u^B_x}{\sum_{y\in A\cap B}u^B_y}\cdot\frac{\Big(\sum_{y\in A\cap B}u^B_y\Big)^{\lambda}}{\sum_{C:A\cap C\neq\emptyset}\Big(\sum_{z\in A\cap C}u^C_z\Big)^{\lambda}}.
\end{equation}

We prove the above by two steps.

\medskip
\noindent\textbf{Step 1.} There is some function $\sigma:\mathds{R}^N_{++}\to \mathds{R}^N$ such that $U\in \mathds{R}^N_{++}$ is a fixed point of $\sigma$ iff $p$ admits the representation (15) with respect to $U$.\bigskip

Let $\tilde{U}=\{\tilde{u}^A_x\}_{A\in\mathscr{A}, x\in A}$. Let us define the following mapping for each $U, \tilde{U}$, and $\lambda$: 
\begin{eqnarray*}
q(x, A|U, \tilde{U}, \lambda)&\equiv&\frac{u^A_x}{\sum_{y\in A} u^A_y}\cdot\frac{\Big(\sum_{y\in A}\tilde{u}^A_y\Big)^{\lambda}}{\sum_{C:A\cap C\neq\emptyset}\Big(\sum_{z\in A\cap C}\tilde{u}^C_z\Big)^{\lambda}}\\
&+&\sum_{B: x\in B, A\neq B}\frac{\tilde{u}^B_x}{\sum_{y\in A\cap B}\tilde{u}^B_y}\cdot\frac{\Big(\sum_{y\in A\cap B}\tilde{u}^B_y\Big)^{\lambda}}{\sum_{C:A\cap C\neq\emptyset}\Big(\sum_{z\in A\cap C}\tilde{u}^C_z\Big)^{\lambda}}.\\
\end{eqnarray*}
Now note that it is enough for us to find $U$ and $\lambda$ such that $p(x, A)=q(x, A|U, U, \lambda)$. For notational simplicity, let 
\[f^A(\tilde{U}, \lambda)\equiv\frac{\Big(\sum_{y\in A}\tilde{u}^A_y\Big)^{\lambda}}{\sum_{C:A\cap C\neq\emptyset}\Big(\sum_{z\in A\cap C}\tilde{u}^C_z\Big)^{\lambda}}\]
and 
\[g^A_x(\tilde{U}, \lambda)\equiv\sum_{B: x\in B, A\neq B}\frac{\tilde{u}^B_x}{\sum_{y\in A\cap B}\tilde{u}^B_y}\cdot\frac{\Big(\sum_{y\in A\cap B}\tilde{u}^B_y\Big)^{\lambda}}{\sum_{C:A\cap C\neq\emptyset}\Big(\sum_{z\in A\cap C}\tilde{u}^C_z\Big)^{\lambda}}.\]
Then we have $q(x, A|U, \tilde{U}, \lambda)=\frac{u^A_x}{\sum_{y\in A} u^A_y}\cdot f^A(\tilde{U}, \lambda)+g^A_x(\tilde{U}, \lambda)$. For any $M\equiv\{m^A\}_{A\in\mathscr{A}}\in \mathds{R}^{|\mathscr{A}|}_{++}$, let 
\[\sigma^A_x(\tilde{U}, \lambda, M)\equiv m^A\,\frac{p(x, A)-g^A_x(\tilde{U}, \lambda)}{f^A(\tilde{U}, \lambda)}\text{ and }\sigma(\tilde{U}, \lambda, M)\equiv \{\sigma^A_x(\tilde{U}, \lambda, M)\}_{A\in\mathscr{A}, x\in A}.\]
Note that $f^A(\cdot, \lambda, M)$ and $g^A_x(\cdot, \lambda, M)$ are strictly positive and continuous functions on $\mathds{R}^N_{++}$. Hence, $\sigma(\cdot, \lambda, M):\mathds{R}^N_{++}\to \mathds{R}^N$ is continuous for each $(\lambda, M)$.  Moreover,

\begin{eqnarray*}
q(x, A|\sigma(\tilde{U}, \lambda, M), \tilde{U}, \lambda)&=&\frac{\sigma^A_x(\tilde{U}, \lambda, M)}{\sum_{y\in A} \sigma^A_y(\tilde{U}, \lambda, M)}\cdot f^A(\tilde{U}, \lambda)+g^A_x(\tilde{U}, \lambda)\\
&=&\frac{m^A\,\frac{p(x, A)-g^A_x(\tilde{U}, \lambda)}{f^A(\tilde{U}, \lambda)}}{\sum_{y\in A}m^A\,\frac{p(y, A)-g^A_y(\tilde{U}, \lambda)}{f^A(\tilde{U}, \lambda)}}\cdot f^A(\tilde{U}, \lambda)+g^A_x(\tilde{U}, \lambda)\\
&=&p(x, A)\text{ since }f^A(\tilde{U}, \lambda)=1-\sum_{x\in A}g^A_x(\tilde{U}, \lambda).
\end{eqnarray*}
Therefore, it is enough to find $U\in \mathds{R}^N_{++}$ such that $\sigma(U, \lambda, M)=U$; i.e., a fixed point of $\sigma(\cdot, \lambda, M)$ in $\mathds{R}^N_{++}$. To apply Brouwer's fixed point theorem,\footnote{\noindent\textbf{Brouwer's fixed point theorem:} Let $S\subset \mathds{R}^m$ be convex and compact and let $f:S\to S$ be continuous. Then $f$ has a fixed point; that is, there is $s\in S$ such that $f(s)=s$. For example, see \cite{ok2007real}, p.279.} we shall show that for some $(\lambda, M)$ there is a non-empty, convex, compact set $S\subset R^{N}_{++}$ such that $\sigma(\cdot, \lambda, M)$ is a self-map on $S$; that is, $\sigma(\cdot, \lambda, M):S\to S$.\medskip

\noindent\textbf{Step 2.} For some $\lambda$ and $M$, there is a non-empty, closed, convex set $S\subset R^{N}_{++}$ such that $\sigma(U, \lambda, M)\in S$ for any $U\in S$. Let $p^*\equiv\min\{\min_{B\in\mathscr{A}, y\in B} p(y, B), \frac{1}{|X|}\}>0$ and

\[S\equiv\{U\in\mathds{R}^N_{++}|\sum_{x\in A} u^A_x=1+\frac{|A|}{|X|^2} p^*\text{ and }\sum_{x\in B}u^A_x\le 1-\frac{(p^*)^2}{4}\text{ for any }A, B\text{ with }B\subset A\}.\]

\noindent\textbf{Step 2.1.} $S$ is non-empty. 

We will show that $U\in S$ when $u^A_x=\frac{1}{|A|}+\frac{p^*}{|X|^2}$ for each $A\in\mathscr{A}$ and $x\in X$. First, $\sum_{x\in A} u^A_x=1+\frac{|A|}{|X|^2}p^*$. Second, for any $B\subset A$, 
\[\sum_{x\in B} u^A_x\le (|A|-1)\big(\frac{1}{|A|}+\frac{p^*}{|X|^2}\big)=1-(\frac{1}{|A|}+\frac{p^*}{|X|^2}-\frac{|A|}{|X|^2}p^*)<1-(\frac{1}{|A|}-\frac{|A|}{|X|^2}p^*)\]
\[\le 1-(\frac{1}{|X|}-\frac{|X|}{|X|^2}p^*)=1-\frac{(1-p^*)}{|X|}\le 1-\frac{1}{2|X|}\le 1-\frac{p^*}{2}\le 1-\frac{(p^*)^2}{4}.\]

\medskip
\noindent\textbf{Step 2.2.} $S$ is convex. 

Let
\[S_1\equiv\{U\in\mathds{R}^N_{++}|\sum_{x\in A} u^A_x=1+\frac{|A|}{|X|^2} p^*\}\]
and
\[S_2\equiv\{U\in\mathds{R}^N_{++}|\sum_{x\in B}u^A_x\le 1-\frac{(p^*)^2}{4}\text{ for any }A, B\text{ with }B\subset A\}.\]
Note that $S_1, S_2$ are convex sets. Hence, $S=S_1\cap S_2$ is convex. 

\medskip
\noindent\textbf{Step 2.3.} $S$ is compact.

Note that $S\subset [0, 2]^N$. Hence, $S$ is bounded. Moreover, note that $S_1, S_2$ are closed sets. Hence, $S=S_1\cap S_2$ is closed.

\medskip
\noindent\textbf{Step 2.4.} For some $\lambda$, $f^A(U, \lambda)\in (1-\frac{p^*}{2}, 1)$ for any $U\in S$.\medskip

First, it is immediate that
\[f^A(U, \lambda)=\frac{\Big(\sum_{y\in A} u^A_y\Big)^{\lambda}}{\sum_{C:A\cap C\neq\emptyset}\Big(\sum_{z\in A\cap C}u^C_z\Big)^{\lambda}}<1\text{ when }U\gg 0.\]
Second, in order to show that 
\[f^A(U, \lambda)=\frac{\Big(\sum_{y\in A} u^A_y\Big)^{\lambda}}{\sum_{C:A\cap C\neq\emptyset}\Big(\sum_{z\in A\cap C}u^C_z\Big)^{\lambda}}>1-\frac{p^*}{2},\]
it is enough to prove that 
\[\frac{1}{f^A(U, \lambda)}-1=\frac{\sum_{C:A\cap C\neq\emptyset, C\neq A}\Big(\sum_{z\in A\cap C}u^C_z\Big)^{\lambda}}{\Big(\sum_{y\in A} u^A_y\Big)^{\lambda}}<\frac{p^*}{2}.\]
By the construction of $S$, we have $\Big(\sum_{y\in A} u^A_y\Big)^{\lambda}=\Big(1+\frac{|A|}{|X|^2} p^*\Big)^{\lambda}$.
Moreover, 
\[\sum_{C:A\cap C\neq\emptyset, C\neq A}\Big(\sum_{z\in A\cap C}u^C_z\Big)^{\lambda}=\sum_{B\subset A}\Big(\sum_{y\in B}u^B_y\Big)^{\lambda}+\sum_{C:A\cap C\neq\emptyset, C\not{\subseteq} A}\Big(\sum_{z\in A\cap C}u^C_z\Big)^{\lambda}\]
and by the construction of $S$, 
\[\sum_{B\subset A}\Big(\sum_{y\in B}u^B_y\Big)^{\lambda}=\sum_{B\subset A}\Big(1+\frac{|B|}{|X|^2} p^*\Big)^{\lambda}<2^{|X|}\Big(1+\frac{|A|-1}{|X|^2} p^*\Big)^{\lambda}.\]
Moreover, since $C\cap A\neq C$ whenever $C\nsubseteq A$, by the construction of $S$,
\[\sum_{C:A\cap C\neq\emptyset, C\nsubseteq A}\Big(\sum_{z\in A\cap C}u^C_z\Big)^{\lambda}\le \sum_{C:A\cap C\neq\emptyset, C\nsubseteq A}\Big(1-\frac{(p^*)^2}{4}\Big)^{\lambda}<2^{|X|}.\]
Combining the last two inequalities, we have
\[\frac{1}{f^A(U, \lambda)}-1=\frac{\sum_{C:A\cap C\neq\emptyset, C\neq A}\Big(\sum_{z\in A\cap C}u^C_z\Big)^{\lambda}}{\Big(\sum_{y\in A} u^A_y\Big)^{\lambda}}<\frac{2^{|X|}(1+\frac{|A|-1}{|X|^2} p^*)^{\lambda}+2^{|X|}}{(1+\frac{|A|}{|X|^2} p^*)^{\lambda}}.\]
Let $\lambda^*\equiv\max_{A\in\mathscr{A}}\frac{\log\Big(\frac{2^{|X|+2}}{p^*}\Big)}{\log\Big(\frac{1+\frac{|A|}{|X|^2} p^*}{1+\frac{|A|-1}{|X|^2} p^*}\Big)}$ and $\lambda>\lambda^*$. Then we have $\frac{2^{|X|+2}}{p^*}<\Big(\frac{1+\frac{|A|}{|X|^2} p^*}{1+\frac{|A|-1}{|X|^2} p^*}\Big)^\lambda$. Consequently,
\[\frac{p^*}{2}>\frac{2^{|X|+1}\big(1+\frac{|A|-1}{|X|^2} p^*\big)^{\lambda}}{\big(1+\frac{|A|}{|X|^2} p^*\big)^{\lambda}}>\frac{2^{|X|}\big(1+\frac{|A|-1}{|X|^2} p^*\big)^{\lambda}+2^{|X|}}{\big(1+\frac{|A|}{|X|^2} p^*\big)^{\lambda}}>\frac{1}{f^A(U, \lambda)}-1.\]

\medskip
\noindent\textbf{Step 2.5.} When $\lambda>\lambda^*$, $g^A_x(U, \lambda)\in [0, \frac{p^*}{2})$ for any $U\in S$.\medskip

It is immediate that $g^A_x(U, \lambda)\ge 0$. Moreover, since $\sum_{x\in A} g^A_x(U, \lambda)=1-f^A(U, \lambda)\in (0, \frac{p^*}{2})$ by Step 2.4, $g^A_x(U, \lambda)< \frac{p^*}{2}$.

\medskip
\noindent\textbf{Step 2.6.} Let $m^A=1+\frac{|A|}{|X|^2} p^*$ and $\lambda>\lambda^*$. Then $\sigma(U, \lambda, M)\in S$ for any $U\in S$.

\medskip
To show that $\sigma(U, \lambda, M)\in S$, we shall prove that for any $A\in\mathscr{A}$, (i) $\sigma^A_x(U, \lambda, M)>0$, (ii) $\sum_{x\in A} \sigma^A_x(U, \lambda, M)=1+\frac{|A|}{|X|^2} p^*$, and (iii) $\sum_{y\in B} \sigma^A_y(U, \lambda, M)\le 1-\frac{(p^*)^2}{4}$ for any $B\subset A$.

\bigskip
\noindent\textbf{Step 2.6.(i).} $\sigma^A_x(U, \lambda, M)>\frac{m^A\,p^*}{2}$.
\bigskip

By Step 2.5 and the definition of $p^*$, we have $p(x, A)\ge p^*$ and $\frac{p^*}{2}>g^A_x(U, \lambda)$. Therefore, since $f^A(U, \lambda)<1$, $\sigma^A_x(U, \lambda, M)=m^A\,\frac{p(x, A)-g^A_x(U, \lambda)}{f^A(U, \lambda)}>m^A\,\frac{\frac{p^*}{2}}{f^A(U, \lambda)}>\frac{m^A\,p^*}{2}$.

\bigskip
\noindent\textbf{Step 2.6.(ii).} $\sum_{x\in A} \sigma^A_x(U, \lambda, M)=1+\frac{|A|}{|X|^2} p^*$.
\bigskip

Since $f^A(U, \lambda)=1-\sum_{x\in A} g^A_x(U, \lambda)$ and by the definition of $m^A$,\\ $\sum_{x\in A}\sigma^A_x(U, \lambda, M)=$$\sum_{x\in A} m^A\,\frac{p(x, A)-g^A_x(U, \lambda)}{f^A(U, \lambda)}=m^A\,\frac{1-\sum_{x\in A} g^A_x(U, \lambda)}{f^A(U, \lambda)}=m^A=1+\frac{|A|}{|X|^2} p^*$.

\bigskip
\noindent\textbf{Step 2.6.(iii).} $\sum_{y\in B} \sigma^A_y(U, \lambda, M)\le 1-\frac{(p^*)^2}{4}$ for any $B\subset A$.
\bigskip

Suppose $x\in A\setminus B$. Then by Step 2.6.(i),
\[\sum_{y\in B} \sigma^A_y(U, \lambda, M)\le \sum_{y\in A}\sigma^A_y(U, \lambda, M)-\sigma^A_x(U, \lambda, M)=m^A-\sigma^A_x(U, \lambda, M)<m^A-\frac{m^A\,p^*}{2}.\]
Finally,
\[m^A-\frac{m^A\,p^*}{2}=(1-\frac{p^*}{2})(1+\frac{|A|}{|X|^2} p^*)\le (1-\frac{p^*}{2})(1+\frac{|X|}{|X|^2} p^*)\le (1-\frac{p^*}{2})(1+\frac{1}{2} p^*)=1-\frac{(p^*)^2}{4}.\]

To sum up, Step 1 shows that $p$ admits the representation (15) with respect to $U$ if and only if $U$ is a fixed point of $\sigma$. Step 2 shows that $\sigma$ has a fixed point by Brouwer's fixed point theorem. Therefore, by Steps 1-2 and Lemma 1, any $p$ is an unrestricted cross-nested logit.

\subsection{Proof of \autoref{NMPprop}} 

We use standard strong consistency results for M-estimators (e.g., see p. 2121-2 of Newey and McFadden (1994)). To obtain $\hat{\mathcal{X}}\xrightarrow{a.s.}\mathcal{X}^*$, since the set of all nest structures is finite, we only need to prove that $\mathcal{X}^*$ is the unique minimizer of $D^*$ and $D\xrightarrow{a.s.} D^*$.  

We first simplify the calculation of $D_1(\mathcal{Y})$. Let $\bar{\epsilon}_{B, A}\equiv\frac{\sum_{a\in A\cap B}\epsilon_{a, A}}{\overline{p}(A\cap B, A)}$. Then $r_A(a, b)=\frac{\overline{p}(a, A)+\epsilon_{a, A}}{\overline{p}(b, A)+\epsilon_{b, A}}=\bar{r}_A(a, b)\,\frac{1+\bar{\epsilon}_{a, A}}{1+\bar{\epsilon}_{b, A}}$. Let
\[\zeta_{a, b, A, B}\equiv\log\Big(\frac{1+\bar{\epsilon}_{a, A}}{1+\bar{\epsilon}_{b, A}}\Big{/}\frac{1+\bar{\epsilon}_{a, B}}{1+\bar{\epsilon}_{b, B}}\Big)\text{ and }\delta_{a, b, A, B}\equiv \log\Big(\frac{\bar{r}_A(a, b)}{\bar{r}_B(a, b)}\Big).\]
Then
\[\log\Big(\frac{r_A(a, b)}{r_B(a, b)}\Big)=\delta_{a, b, A, B}+\zeta_{a, b, A, B}.\]

We now simplify the calculation of $D_2(\mathcal{Y})$. Note that
\[r_A(Y, Y')=\frac{p(A\cap Y, A)}{p(A\cap Y', A)}=\frac{\overline{p}(A\cap Y, A)+\sum_{a\in A\cap Y}\epsilon_{a, A}}{\overline{p}(A\cap Y', A)+\sum_{a\in A\cap Y'}\epsilon_{a, A}}=\bar{r}_A(Y, Y')\,\frac{1+\bar{\epsilon}_{Y, A}}{1+\bar{\epsilon}_{Y', A}}.\]
Similarly, let 
\[\zeta_{Y, Y', A, B}\equiv\log\Big(\frac{1+\bar{\epsilon}_{Y, A}}{1+\bar{\epsilon}_{Y', A}}\Big{/}\frac{1+\bar{\epsilon}_{Y, B}}{1+\bar{\epsilon}_{Y', B}}\Big)\text{ and }\delta_{Y, Y', A, B}\equiv \log\Big(\frac{\bar{r}_A(Y, Y')}{\bar{r}_B(Y, Y')}\Big).\]
Then
\[\log\big(\frac{r_A(Y, Y')}{r_B(Y, Y')}\big)=\delta_{Y, Y', A, B}+\zeta_{Y, Y', A, B}.\]
Let $N_1(\mathcal{Y})=\sum_{Y\in\mathcal{Y}}|\{(A, B, a, b)| a, b\in A\cap B\cap Y\}|$ and $N_2(\mathcal{Y})=\sum_{Y, Y'\in\mathcal{Y}} |\{(A, B)| A\cap Y=B\cap Y,\,A\cap Y'=B\cap Y'\}|$ and take any $M>\max_{\mathcal{Y}}N_1(\mathcal{Y}), \max_{\mathcal{Y}} N_2(\mathcal{Y})$. Hence, 
\[D(\mathcal{Y})=\frac{\sum_{Y\in \mathcal{Y}}\sum_{A, B\in\mathscr{A}, a, b\in A\cap B\cap Y}
 (\delta_{a, b, A, B}+\zeta_{a, b, A, B})^2}{N_1(\mathcal{Y})}\]
 \[+\frac{\sum_{Y, Y'\in\mathcal{Y}} \sum_{A, B\in\mathscr{A}: A\cap Y=B\cap Y,\, A\cap Y'=B\cap Y'}(\delta_{Y, Y', A, B}+\zeta_{Y, Y', A, B})^2}{N_2(\mathcal{Y})}.\]

We then show that $\mathcal{X}^*$ is the unique minimizer of $D^*$. Since $D^*(\mathscr{X}^*)=0$, we shall show that  $D^*(\mathcal{Y})>0$ for any nest structure $\mathcal{Y} \neq \mathcal{X}^*$. It is enough to consider the following two cases.\medskip

\noindent\textbf{Case 1.} $\mathcal{Y}$ is a partition of $X$ such that there are $a\in X_i$ and $b\in X_j$ such that $a, b\in Y$ for some $Y\in\mathcal{Y}$.\medskip

By Assumption 1, there are $A, B$ with $\bar{r}_A(a, b)\neq \bar{r}_B(a, b)$; i.e., $\delta_{a, b, A, B}\neq 0$. Then we have $D^*(\mathcal{Y})\ge D^*_1(\mathcal{Y})>(\delta_{a, b, A, B})^2/M>0$.

\smallskip
\noindent\textbf{Case 2.} $\mathcal{Y}$ is a partition of $X$ such that for any $Y\in\mathcal{Y}$, $Y\subseteq X_i$ for some $i$, and $Y'\subset X_j$ for some $Y'\in\mathcal{Y}$ and $j$.\medskip

Take any $Y, Y'$ such that $Y\subset X_i$ and $Y'\subseteq X_j$. By Assumption 1, there are $A, B$ such such that $\bar{r}_A(Y, Y')\neq \bar{r}_B(Y, Y')$, $A\cap Y=B\cap Y$, and $A\cap Y'=B\cap Y'$. That is, $\delta_{Y, Y', A, B}\neq 0$.  Then we have $D^*(\mathcal{Y})\ge D^*_2(\mathcal{Y})\ge (\delta_{Y, Y', A, B})^2/M>0$.

\medskip
We finally show that $D(\mathcal{Y})\xrightarrow{a.s.}  D^*(\mathcal{Y})$ for every $\mathcal{Y}$. We have \begin{eqnarray*}D(\mathcal{Y})-D^*(\mathcal{Y})&=&\sum \frac{\big((\delta_{a, b, A, B}+\zeta_{a, b, A, B})^2-(\delta_{a, b, A, B})^2\big)}{N_1(\mathcal{Y})}+\sum\frac{\big((\delta_{Y, Y', A, B}+\zeta_{Y, Y', A, B})^2-(\delta_{Y, Y', A, B})^2\big)}{N_2(\mathcal{Y})}\\
&=&\sum \frac{\zeta_{a, b, A, B}\,(2\,\delta_{a, b, A, B}+\zeta_{a, b, A, B})}{{N_1(\mathcal{Y})}}+\sum\frac{\zeta_{Y, Y', A, B}\,(2\,\delta_{Y, Y', A, B}+\zeta_{Y, Y', A, B})}{{N_2(\mathcal{Y})}}\xrightarrow{a.s.} 0
\end{eqnarray*}
since $\zeta_{a, b, A, B}\xrightarrow{a.s.} 0$, $\zeta_{Y, Y', A, B}\xrightarrow{a.s.} 0$, and $\delta_{a, b, A, B}$ and $\delta_{Y, Y', A, B}$ are constants.

\subsection{Proof of \autoref{Reduceprop}} 

Take any $a, b\in X$. As we showed in the proof of Proposition 6, we have $\log\big(r_A(a, b)/r_B(a, b)\big)=\delta_{a, b, A, B}+\zeta_{a, b, A, B}$. If $a, b\in X_i$ for some $i$, then
\[d(a, b)=\frac{\sum_{A, B\in\mathscr{A}: a, b\in A\cap B} \zeta^2_{a, b, A, B}}{|\{(A, B, a, b)| a, b\in A\cap B\}|}\xrightarrow{a.s.} 0.\]

If $a\in X_i$ and $b\in X_j$ for some $i, j$ with $i\neq j$, then by Assumption 1, there are $A^*, B^*$ with $\bar{r}_{A^*}(a, b)\neq \bar{r}_{B^*}(a, b)$; i.e., $\delta_{a, b, A^*, B^*}\neq 0$. Hence, 
\[d(a, b)>\frac{(\delta_{a, b, A^*, B^*}+\zeta_{a, b, A^*, B^*})^2}{M}\ge \frac{(\delta_{a, b, A^*, B^*})^2}{2M}\text{ almost surely}.\]
Let
\[\epsilon^*\equiv \min_{a', b', A', B': \delta_{a', b', A', B'}\neq 0}\frac{(\delta_{a', b', A', B'})^2}{2M}.\]
Then by the previous inequality, $d(a, b)>\epsilon^*$ almost surely.  Since $\epsilon^*>0$, there is $\bar{N}$ such that for any $N^*>\bar{N}$,\[\max_{i}\max_{a, b\in X_i} d(a, b)<\epsilon^*<\min_{i<j}\min_{a'\in X_i, b'\in X_j} d(a', b')\,\text{ with probability one.}\]

\subsection{Proofs of \autoref{NestReduceprop} and Corollary 2} 

\noindent\textbf{Proof of \autoref{NestReduceprop}.} Since $X$ and $\mathscr{A}$ are finite, there is a set $\{d_1, \ldots, d_m\}$ of positive real numbers such that $\max_{a, b} d(a, b)=d_m>\ldots>d_2>d_1=\min_{a, b} d(a, b)\ge 0$ and for any $a', b'\in X$, $d(a', b')=d_s$ for some $s\le m$. Hence, it is immediate that $|\mathscr{X}^*|\le |X|^2$.

Let us prove $|\mathscr{X}^*|\le |X|$ by induction on the number of alternatives $|X|$. When $|X|=2$, we have $|\mathscr{X}^*|\le |\mathscr{X}|=2=|X|$. Suppose that the hypothesis is true for any set $X$ with $|X|=k$. We shall prove that this also holds for all sets $X$ with $|X|=k+1$.

Take a set $X$ and suppose $|X|=k+1$. Let $a^*, b^*$ be elements of $X$ such that $d_1=d(a^*, b^*)$. Hence, for any $\epsilon\in (d_1, d_m]$, $a^*\sim_\epsilon b^*$. In other words, $a^*$ and $b^*$ belong to the same nest for any partition $\mathcal{X}_\epsilon$ with $\epsilon\in (d_1, d_m]$. If $d_1=d(a^*, b^*)=d_m=\max_{a, b\in X} d(a, b)$, then we obtain the desired result since $|\mathscr{X}^*|=1$. Suppose that $d(a^*, b^*)<d_m$. Then, without loss of generality, we can assume that $d_m=\max_{a, b\in X'} d(a, b)$ where $X'=X\setminus\{b^*\}$.  Because $|X'|\le k$, it follows from the induction assumption that there are at most $k$ different partitions in $\mathscr{X'}^*=\{\mathcal{X}'_\epsilon\}_{\epsilon\in [0, d_m]}$. 

Now let us consider $\mathscr{X}^*$. When $\epsilon\in (d_1, d_m]$, adding $b^*$ to $X'$ does not increase the number of distinct partitions in $\{\mathcal{X}'_\epsilon\}_{\epsilon\in [0, d_m]}$ since $a^*$ and $b^*$ must belong to the same nest. Note that if $\sim_\epsilon$ is transitive on $X$, then it is also transitive on $X'$. Hence, adding $b^*$ to $X'$ does not extend the set of $\epsilon$ such that $\sim_{\epsilon}$ is transitive. Therefore, there is at most one new partition when $b^*$ is added to $X'$. Therefore, by induction, $|\mathscr{X}^*|\le |X|$.

\bigskip
\noindent\textbf{Proof of Corollary 2.} In the proof of \autoref{NestReduceprop}, we show that there is $\epsilon^*>0$ such that
\[\max_{i}\max_{a, b\in X_i} d(a, b)<\epsilon^*<\min_{i<j}\min_{a'\in X_i, b'\in X_j} d(a', b')\,\text{ almost surely.}\]
Hence $\mathcal{X}_{\epsilon^*}=\mathcal{X}^*$ almost surely. Therefore, $\mathcal{X}^*\in\mathscr{X}^*$ almost surely. By Proposition 6, $\hat{\mathcal{X}}\xrightarrow{a.s.}\mathcal{X}^*$. Therefore, since $\mathcal{X}^*\in\mathscr{X}^*$ almost surely and $\mathscr{X}^*\subset \mathscr{X}$, we have $\hat{\mathcal{X}}^*\xrightarrow{a.s.}\mathcal{X}^*$.

\section{Additional Results}

\subsection{Regularity, Increasing NSC, and The Similarity Effect}\label{monNSC}

\autoref{runlprop} shows that regularity has important behavioral implications for nested logit. 
In this section, we study the implications of regularity for general NSC. We show that there is a deep connection between regularity, increasing NSCs, and the similarity effect. To clarify the implications of regularity, we divide regularity into two logically independent axioms (as we did for IIA).

\begin{axm}[Dissimilar Regularity]\label{DR} For any $ A \in\mathscr{A}$,  $x \in A$, and $y \in  X$, 
  \[p(x,A\cup y) \le p(x, A)\text{ when }x \nsim_p y.\]
\end{axm}

\begin{axm}[Similar Regularity]\label{SR} For any $ A \in\mathscr{A}$,  $x \in A$, and $y \in  X$,
  \[p(x,A\cup y) \le p(x, A)\text{ when }x\sim_p y.\]
\end{axm}

The first axiom, \nameref{DR}, says that regularity should hold when $x$ and $y$ are revealed dissimilar, while the second axiom, \nameref{SR}, requires regularity when $x$ and $y$ are revealed similar. It is immediate that the joint assumption of Dissimilar Regularity and Similar Regularity is equivalent to regularity.

The first result shows that Similar Regularity is closely related to the similarity effect. In fact, the similarity effect implies Similar Regularity in a setting that is more general than NSC. 

\begin{prps}\label{SRsim} For any stochastic choice function $p$, if $\sim_p$ is transitive, then the \textbf{similarity effect} implies \textbf{\nameref{SR}}. \end{prps}

The intuition behind this result is quite simple. If an alternative $y$ is introduced and it is similar to some existing alternative $x$, the similarity effect requires that $y$ hurts $x$ more than it hurts anything that it is not similar to. Thus, the probability of $x$ must decrease. Note that this result only requires that $\sim_p$ is transitive; it does not rely on the structure of NSC. Consequently, it is difficult to explain both the similarity effect and violations of regularity. Further, we can show that under some richness condition, \nameref{SR} will imply a weak version of the similarity effect, and thus the similarity effect is essentially equivalent to \nameref{SR}. 

The second result shows that \nameref{DR} is equivalent to a mild but behaviorally important restriction on $v$: monotonicity in the size of the nest. We say $p$ is an \textbf{increasing NSC} if $v(A) \ge v(B)$ for any $i\le K$ and nonempty sets $A, B\subseteq X_i$ with $B\subseteq A$.

\begin{prps}\label{incNSC} A nondegenerate NSC $p$ satisfies \textbf{\nameref{DR}} if and only if it is an \textbf{increasing NSC}. 
\end{prps}

Increasing NSC are interesting because they subsume many of the models in the literature, including nested logit. However, increasing NSC are incompatible with certain violations of regularity, such as choice overload. In fact, Propositions 7 and 8 imply that increasing NSC cannot allow the similarity effect and violations of regularity simultaneously. Thus, non-increasing NSC (e.g., the menu-dependent substitutability example from section \ref{bnl}) are of independent interest.

\subsection{Alternative Axiomatization of NSC}

In this section, we provide an alternative axiomatic characterization of NSC in which characterizing axioms do not rely on our revealed similarity relation $\sim_p$. To characterize NSC, we ``divide" ISA into two axioms. 

\begin{axm}[ISA-1]\label{ISSA} For any $A\in\mathscr{A}$, $a, b\in A$, and $x\not\in A$, 
\[\frac{p(a, \{a, x\})}{p(x, \{a, x\})}=\frac{p(a, \{a, b, x\})}{p(x, \{a, b, x\})}\text{ and }\frac{p(b, \{b, x\})}{p(x, \{b, x\})}=\frac{p(b, \{a, b, x\})}{p(x, \{a, b, x\})}\,\,\,\Longrightarrow\,\,\,\frac{p(a, A)}{p(b, A)}=\frac{p(a, A\cup x)}{p(b, A\cup x)}.\]
\end{axm}

Note that the above axiom is essentially identical to the first part of ISA with $a\sim_p x$ and $b\sim_p x$. 

\begin{axm}[ISA-2]\label{ISDA} For any $A, B, C\in\mathscr{A}$, $a\in A\cap B, b\in A\cap C$, and $x\in B\cap C$, 
\[\frac{p(a, \{a, x\})}{p(x, \{a, x\})}\neq \frac{p(a, B)}{p(x, B)}\text{ and }\frac{p(b, \{b, x\})}{p(x, \{b, x\})}\neq \frac{p(b, C)}{p(x, C)}\,\,\,\,\Longrightarrow\,\,\,\,\frac{p(a, A)}{p(b, A)}=\frac{p(a, A\cup x)}{p(b, A\cup x)}.\]
\end{axm}

Similarly, the above axiom is essentially identical to the second part of ISA with $a\not\sim_p x$ and $b\not\sim_p x$. 

We also need to strengthen our notion of nondegeneracy as follows. The NSC $p$ with $(v, u, \{X_i\}^N_{i=1})$ is \textbf{strict} if for any $A_i\subset X_i$ with $a\in A_i$ and $x\in X_i\setminus A_i$,
\[\text{if }\frac{u(a)+u(x)}{u(x)}=\frac{v(\{a, x\})}{v(x)}, \text{ then }\frac{\sum_{a'\in A_i}u(a')+u(x)}{\sum_{a'\in A_i}u(a')}=\frac{v(A_i\cup x)}{v(A_i)}.\]

\begin{thm}\label{AltNSC} Let $p$ be a stochastic choice function with at least three alternatives that are dissimilar to each other. Then $p$ satisfies \autoref{ISSA} and \autoref{ISDA} if and only if it is a strict nondegenerate NSC.
\end{thm}

\subsection{Three-Step Nested Stochastic Choice}\label{3nsc}

\autoref{NSCthm} shows that all two-level nested logit models are in fact special cases of NSC and are characterized by a strong notion of categorical similarity.  A natural question is, can we capture more complex substitution patterns through a more general notion of similarity?  Put another way, can we allow for contextual or comparative similarity? 

Consider a decision maker who is choosing between wines and beers. As in NSC, it is natural to think that there are three nests: one for white wines, one for red wines, and one for beers.  Intuitively, the wines are ``more similar'' to each other than they are to the beers. This can be captured though an intermediate step in which, before deciding between red or white wines, the decision maker decides between wines and beers.  After deciding between wine and beer, the consumer decides between different styles of wine (red vs. white), and then selects a specific one to consume.  This can be represented through a three-level nested structure, which we refer to as a 3-step NSC.  In this section, we show that we can capture such complex relationships through the introduction of a second similarity relation and a generalization of our main axiom. 

Formally, any $3$-step NSC consists of a nesting structure (tree) and conditional Luce rules.  

\begin{defn}[3-step NSC] A stochastic choice function $p$ is a \textbf{3-Step Nested Stochastic Choice} if there exist a partition $X_1, \ldots, X_K$ of $X$, a partition $X^1_k, \ldots, X^{q_k}_k$ of $X_k$ for each $k\le K$, and functions $u: X\to\mathds{R}_{++}, w: \bigcup^K_{i=1} 2^{X_i}\to\mathds{R}_{+}$, and $v:\bigcup^K_{k=1} \bigcup^{q_k}_{l=1} 2^{X^l_k}\to\mathds{R}_{+}$ with $w(\emptyset)=v(\emptyset)=0$ such that for any $A\in\mathscr{A}$ and $x\in A\cap X^j_k$,
\[p(x, A)=\frac{u(x)}{\sum_{y\in A\cap X^j_k}u(y)}\cdot \frac{v\big(A\cap X^j_k\big)}{\sum^{q_k}_{l=1}v\big(A\cap X^l_k\big)}\cdot \frac{w\big(A\cap X_k\big)}{\sum^K_{i=1}w\big(A\cap X_{i}\big)}.\]
\end{defn}

We now introduce a secondary notion of similarity which applies to alternatives that are not categorically similar, but satisfy IIA in the presence of mutually dissimilar alternatives. 

\begin{defn} For any $a, b\in X$, we say $a$ and $b$ are \textbf{approximately revealed similar}, denoted by $a\bowtie_p b$, if $a\not\sim_p b$ and 
\[\frac{p(a, A)}{p(b, A)}=\frac{p(a, A\cup x)}{p(b, A\cup x)}\text{ for any $A\in\mathscr{A}$ and $x\notin A$ with }x\not\sim_p a\text{ and }x\not\sim_p b.\]
We write $a\simeq_p b$ if either $a\sim_p b$ and $a\bowtie_p b$.
\end{defn}

It is crucial to note that our approximately revealed similar relation requires that $a$ and $b$ are not categorically similar. Thus we have two distinct ``layers'' of similarity; $\bowtie_p$ does not include $\sim_p$ as a sub-relation.  In terms of our drink example, all the red wines are categorically similar (related through $\sim_p$), while red and white wines are approximately similar (related through $\bowtie_p$), as IIA will hold between them when a beer is introduced but not if another wine were introduced.  Hence this second layer delineates the ``intermediate'' nests in the tree and captures aspects of context-dependent similarity. Consequently, this approach distinguishes between fundamental and contextual similarity. 

We now introduce a generalization of our main axiom, to characterize 3-step NSC.

\begin{axm}[Generalized Independence of Symmetric Alternatives]\label{GISA} For any $A\in\mathscr{A}$, $a, b\in A$, and $x \notin A$,

\begin{equation*}
\begin{aligned}[c]
&\,a\sim_p x\text{ and }b\sim_p x,\\
&\, a\bowtie_p x\text{ and }b\bowtie_p x,\\
&\,\,\,\,\,\,\,\,\,\,\,\,\,\,\,\,\text{ or }\\
&\, a\not\simeq_p x\text{ and }b\not\simeq_p x
\end{aligned}
\qquad\Longrightarrow\qquad
\begin{aligned}[c]
\frac{p(a, A)}{p(b, A)}=\frac{p(a, A\cup x)}{p(b, A\cup x)}.
\end{aligned}
\end{equation*}
\end{axm}

Recall that the second part of Independence of Symmetric Alternatives requires that $\frac{p(a, A)}{p(b, A)}=\frac{p(a, A\cup x)}{p(b, A\cup x)}$ when $a\not\sim_p x$ and $b\not\sim_p x$. However, the second part of Generalized Independence of Symmetric Alternatives requires that $\frac{p(a, A)}{p(b, A)}=\frac{p(a, A\cup x)}{p(b, A\cup x)}$ when either $a\bowtie_p x$ and $b\bowtie_p x$ or $a\not\simeq_p x$ and $b\not\simeq_p x$. Hence, Generalized Independence of Symmetric Alternatives relaxes the second part of Independence of Symmetric Alternatives. 

Lastly, we need a consistency condition to hold between the similarity relations. 

\begin{axm}[Consistency of Revealed Similarities] For any $x, y, x'\in X$, if $x\sim_p x'$, then
\[x\bowtie_p y\text{ if and only if } x'\bowtie_p y.\]
\end{axm}

In the language of our drink example, $y$ is a white wine and $x,x'$ are two red wines. Since $x\sim_p x'$, it must be the case that if a white wine $y$ is approximately similar to some red wine $x$, then it is approximately similar to any other red wine $x'$.

\begin{thm}\label{3NSCthm} Consider a stochastic choice function $p$ such that there are $a, b, c$ with $a\not\simeq_p b$, $b\not\simeq_p c$, and $a\not\simeq_p c$. Suppose, for any $x\in X$, there are $y, z\in X$ such that $x\bowtie_p y$, $y\bowtie_p z$, and $x\bowtie_p z$. If $p$ satisfies \textbf{Generalized Independence of Symmetric Alternatives and Consistency of Revealed Similarities}, then it is a \textbf{3-step NSC}.\footnote{Indeed, just as in Theorem 1, the necessity direction also holds when the appropriate nondegeneracy condition is imposed on $p$.}
\end{thm}

The major insight from this result is that multi-step NSC is characterized by revealing multiple, layered similarity relations, and then imposing a generalization of our key axiom. Just as our similarity relation identifies endogenous nests, this secondary relation identifies endogenous, intermediate nests. Thus, \autoref{3NSCthm} shows that we may identify an endogenous tree structure. 

Multi-level nested logit models have been applied to many situations. Most famously, \cite{goldberg1995product} uses a multi-level nested logit to study automobile demand.   It is well known that the ``order'' in which the tree-structure of nests is specified matters for estimates. Our approach reveals the entire, endogenous tree, and so the ``order'' is also recovered: $\bowtie_p$ captures upper nests and $\sim_p$ captures lower nests.  It is straightforward to see how our approach could be extended  to characterize an $N$-Step NSC.

\subsection{Remaining Proofs}

\subsubsection{Proof of \autoref{SRsim}}

As discussed in the proof of Theorem 1, when $\sim_p$ is transitive, there is a partition $\{X_i\}^K_{i=1}$ such that for any $x, y\in X$, $x\sim_p y$ if and only if $x, y\in X_i$ for some $i\le K$. Suppose that $p$ satisfies the similarity effect; that is, for any $A\in\mathscr{A}$ and $a, a'\in X_i$ and $b\in X_j$ with $a, b\in A$ and $a'\not\in A$,
\[\frac{p(a,A)}{p(b,A)} > \frac{p(a,A\cup a')}{p(b,A\cup a')}.\]

\noindent\textbf{Step 1.} For any $A\in\mathscr{A}$ and $a'\in X_i\setminus A$ with $A\cap X_i\neq \emptyset$, $p(A\cap X_i, A)> p(A\cap X_i, A\cup a')$.\medskip

The similarity effect implies that for any $b\in A\setminus X_i$ and $a\in A\cap X_i$,
\begin{equation}\label{SimEq}
p(b,A\cup a')\,p(a,A)> p(b,A)\,p(a,A\cup a').
\end{equation}

Let us first add Inequality (\ref{SimEq}) across all $a\in A\cap X_i$. Then we have 
\[p(b, A\cup a')\,p(A\cap X_i,A)> p(b, A)\,p(A\cap X_i,A\cup a').\]
Let us add again the above inequality across all $b\in A\setminus X_i$. Then we obtain 
\[\big(1-p(A\cap X_i, A\cup a')-p(a', A\cup a')\big)\,p(A\cap X_i,A)> \big(1-p(A\cap X_i, A)\big)\,p(A\cap X_i,A\cup a').\]
The above inequality implies 
\[\big(1-p(A\cap X_i, A\cup a')\big)\,p(A\cap X_i,A)> \big(1-p(A\cap X_i, A)\big)\,p(A\cap X_i,A\cup a');\]
equivalently, $p(A\cap X_i,A)> p(A\cap X_i,A\cup a')$.

\smallskip
\noindent\textbf{Step 2.} For any $A\in\mathscr{A}$, $a\in A\cap X_i$, and $a'\in X_i\setminus A$, $p(a, A)>p(a, A\cup a')$.\medskip

For any $\tilde{a}\in A\cap X_i$, we have $\frac{p(a, A\cup a')}{p(\tilde{a}, A\cup a')}=\frac{p(a, A)}{p(\tilde{a}, A)}$. Let us add the equality $p(a, A\cup a')\,p(\tilde{a}, A)=p(a, A)\,p(\tilde{a}, A\cup a')$ for all $\tilde{a}\in A\cap X_i$. Then we obtain $p(a, A\cup a')\,p(A\cap X_i, A)=p(a, A)\,p(A\cap X_i, A\cup a')$. Finally, since $p(A\cap X_i, A\cup a')<p(A\cap X_i, A)$, we need to have $p(a, A)>p(a,A\cup a')$.

\subsubsection{Proof of \autoref{incNSC}}

Let $p$ be a nondegenerate NSC with $(v, u, \{X_i\}^K_{i=1})$.

\noindent\textbf{Sufficiency.} Take any $A_j\subset X_j$ and $y\in X_j\setminus A_j$. Take any $x \in X_i $ with $i\neq j$. Since $x\not\sim_p y$, by Dissimilar Regularity, we have \[p(x,A_j\cup x\cup y)=
\frac{v(x)}{v(x) + v(A_j\cup y)}\le p(x, A_j\cup x)=
\frac{v(x)}{v(x) + v(A_j)}\text{ iff }v(A_j)\le v(A_j\cup y).\] 

\noindent\textbf{Necessity.} Suppose $v$ is increasing. Take any $A\in\mathscr{A}$, $x\in A$, and $y\not\in A$ with $x\not\sim_p y$. Therefore, $x\in X_i$ and $y\in X_j$ for some $i, j$ with $i\neq j$. Since $v(A\cap X_j\cup y)\ge v(A\cap X_j)$, we have  
\[p(x, A\cup y)=\frac{p(x, A\cap X_i)\,v(A\cap X_i)}{v(A_j\cup y)+\sum_{k\neq j} v(A\cap X_k)}\le 
p(x, A)=\frac{p(x, A\cap X_i)\,v(A\cap X_i)}{v(A_j)+\sum_{k\neq j} v(A\cap X_k)}.\]

\subsubsection{Proof of \autoref{AltNSC}}

\textbf{Sufficiency.} Take any $A\in\mathscr{A}$, $a, b\in A$, and $x\not\in A$. Suppose $a\sim_p x$ and $b\sim_p x$. Then we have $\frac{p(a, \{a, x\})}{p(x, \{a, x\})}=\frac{p(a, \{a, b, x\})}{p(x, \{a, b, x\})}$ and $\frac{p(b, \{b, x\})}{p(x, \{b, x\})}=\frac{p(b, \{a, b, x\})}{p(x, \{a, b, x\})}$. Hence, by Axiom 8, $\frac{p(a, A)}{p(b, A)}=\frac{p(a, A\cup x)}{p(b, A\cup x)}$. Suppose $a\not\sim_p x$ and $b\not\sim_p x$. Then there are $B, C$ such that $\frac{p(a, \{a, x\})}{p(x, \{a, x\})}\neq \frac{p(a, B)}{p(x, B)}$ and $\frac{p(b, \{b, x\})}{p(x, \{b, x\})}\neq \frac{p(b, C)}{p(x, C)}$. Hence, by Axiom 9, $\frac{p(a, A)}{p(b, A)}=\frac{p(a, A\cup x)}{p(b, A\cup x)}$. Therefore, ISA is satisfied. Hence, by Theorem 1, $p$ is a nondegenerate NSC with some $(v, u, \{X_i\}^K_{i=1})$ where $X/\sim_p=\{X_i\}^K_{i=1}$.

To show the strictness, take any $a, x\in X_i$ such that $\frac{u(a)+u(x)}{u(x)}=\frac{v(\{a, x\})}{v(x)}$. Then for any $b\in X_j$, $\frac{p(b, \{b, x\})}{p(x, \{b, x\})}=\frac{p(b, \{a, b, x\})}{p(x, \{a, b, x\})}$. Since $\frac{p(a, \{a, x\})}{p(x, \{a, x\})}=\frac{p(a, \{a, b, x\})}{p(x, \{a, b, x\})}$, by Axiom 8, we have $\frac{p(a, A)}{p(b, A)}=\frac{p(a, A\cup x)}{p(b, A\cup x)}$. By NSC, we have
\[\frac{\frac{u(a)}{\sum_{a'\in A\cap X_i}u(a')} v(A\cap X_i)}{\frac{u(b)}{\sum_{b'\in A\cap X_j}u(b')} v(A\cap X_j)}=\frac{\frac{u(a)}{\sum_{a'\in A\cap X_i}u(a')+u(x)} v(A\cap X_i\cup x)}{\frac{u(b)}{\sum_{b'\in A\cap X_j}u(b')} v(A\cap X_j)}\]
equivalently,
\[\frac{\sum_{a'\in A\cap X_i}u(a')+u(x)}{\sum_{a'\in A\cap X_i}u(a')}=\frac{v(A\cap X_i\cup x)}{v(A\cap X_i)}.\]

\noindent\textbf{Necessity.}  Suppose $p$ is the strict nondegenerate NSC with $(v, u, \{X_i\}^K_{i=1})$. By the necessity part of Theorem 1, $p$ satisfies ISA and $a\sim_p b$ if and only if $a, b\in X_i$.

To prove Axiom 8, take any $A\in\mathscr{A}$, $a, b\in A$, and $x\not\in A$ with $\frac{p(a, \{a, x\})}{p(x, \{a, x\})}=\frac{p(a, \{a, b, x\})}{p(x, \{a, b, x\})}$ and $\frac{p(b, \{b, x\})}{p(x, \{b, x\})}=\frac{p(b, \{a, b, x\})}{p(x, \{a, b, x\})}$. We shall prove that $\frac{p(a, A)}{p(b, A)}=\frac{p(a, A\cup x)}{p(b, A\cup x)}$. It is immediate when $a, b\in X_i$. Hence, suppose $a\in X_i$ and $b\in X_j$. If $x\in X_k$, then we have $a\not\sim_p x$ and $b\not\sim_p x$. Consequently, by ISA, we have $\frac{p(b, \{b, x\})}{p(x, \{b, x\})}=\frac{p(b, \{a, b, x\})}{p(x, \{a, b, x\})}$. Hence, suppose now either $x\in X_i$ or $x\in X_j$. Since the role of $a$ and $b$ are symmetric, suppose $a\in X_i$ without loss of generality. By the NSC, $\frac{p(b, \{b, x\})}{p(x, \{b, x\})}=\frac{p(b, \{a, b, x\})}{p(x, \{a, b, x\})}$ implies $\frac{v(b)}{v(x)}=\frac{v(b)}{\frac{u(x)}{u(a)+u(x)}v(\{a, x\})}$; equivalently, $\frac{u(a)+u(x)}{u(x)}=\frac{v(\{a, x\})}{v(x)}$. Since by the strictness, $\frac{\sum_{a'\in A_i}u(a')+u(x)}{\sum_{a'\in A_i}u(a')}=\frac{v(A_i\cup x)}{v(A_i)}$ for any $A_i\subset X_i$.

By the NSC, $\frac{p(a, A)}{p(b, A)}=\frac{p(a, A\cup x)}{p(b, A\cup x)}$ is equivalent to  
\[\frac{\frac{u(a)}{\sum_{a'\in A\cap X_i}u(a')} v(A\cap X_i)}{\frac{u(b)}{\sum_{b'\in A\cap X_j}u(b')} v(A\cap X_j)}=\frac{\frac{u(a)}{\sum_{a'\in A\cap X_i}u(a')+u(x)} v(A\cap X_i\cup x)}{\frac{u(b)}{\sum_{b'\in A\cap X_j}u(b')} v(A\cap X_j)}.\]
The above equality holds by the strictness since it is equivalent to
\[\frac{\sum_{a'\in A\cap X_i}u(a')+u(x)}{\sum_{a'\in A\cap X_i}u(a')}=\frac{v(A\cap X_i\cup x)}{v(A\cap X_i)}.\]

To prove Axiom 9, take any $A, B, C\in\mathscr{A}$, $a\in A\cap B, b\in A\cap C$, and $x\in B\cap C$ with $\frac{p(a, \{a, x\})}{p(x, \{a, x\})}\neq \frac{p(a, B)}{p(x, B)}$ and $\frac{p(b, \{b, x\})}{p(x, \{b, x\})}\neq \frac{p(b, C)}{p(x, C)}$. Then $a\not\sim_p x$ and $b\not\sim_p x$. Hence, by ISA, $\frac{p(a, A)}{p(b, A)}=\frac{p(a, A\cup x)}{p(b, A\cup  x)}$. 
 
\subsubsection{Proof of \autoref{3NSCthm}}

We prove Theorem 6 by four steps.

\noindent\textbf{Step 1.} Note that the first part of Generalized Independence of Symmetric Alternatives is identical to the first part of Independence of Symmetric Alternatives. Hence, by Steps 1-2 of the proof of Theorem 1, $\sim_p$ is reflexive, transitive, and symmetric, we have a partition $X/\sim_p\equiv\{E_i\}^K_{i=1}$ of $X$ such that for any $x_i, x'_i\in E_i$ and $x_j\in E_j$, $x_i\sim_p x'_i$ and $x_i\not\sim_p x_j$. 

\medskip
\noindent\textbf{Step 2.} $\simeq_p$ is transitive.

Take any $x, y, z\in X$ such that $x\simeq_p y$ and $y\simeq_p z$. If $x\sim_p y$ and $y\sim_p z$, then by Step 1, $x\sim_p z$. If $x\sim_p y$ and $y\bowtie_p z$, then by Consistency of Revealed Similarities, $x\bowtie_p z$. If $x\bowtie_p y$ and $y\sim_p z$, then by Consistency of Revealed Similarities, $x\bowtie_p z$. Finally, we consider the case where $x\bowtie_p y$ and $y\bowtie_p z$. 

Suppose $x\not\sim_p z$. Then we shall prove that $x\bowtie_p z$; i.e., for any $A$ and $t\not\in A$ such that $x\not\sim_p t$ and $z\not\sim_p t$, $\frac{p(x, A)}{p(z, A)}=\frac{p(x, A\cup t)}{p(z, A\cup t)}$.

\noindent\textbf{Case 1.} $y\in A$. In this case, we can write $\frac{p(x, A)}{p(z, A)}=\frac{p(x, A)}{p(y, A)}/\frac{p(z, A)}{p(y, A)}$ and $\frac{p(x, A\cup t)}{p(z, A\cup t)}=\frac{p(x, A\cup t)}{p(y, A\cup t)}/\frac{p(z, A\cup t)}{p(y, A\cup t)}$. If $y\not\sim_p t$, $x\bowtie_p y$ implies $\frac{p(x, A)}{p(y, A)}=\frac{p(x, A\cup t)}{p(y, A\cup t)}$ and $y\bowtie_p z$ implies $\frac{p(z, A)}{p(y, A)}=\frac{p(z, A\cup t)}{p(y, A\cup t)}$. Therefore, $\frac{p(x, A)}{p(z, A)}=\frac{p(x, A)}{p(y, A)}/\frac{p(z, A)}{p(y, A)}=\frac{p(x, A\cup t)}{p(z, A\cup t)}=\frac{p(x, A\cup t)}{p(y, A\cup t)}/\frac{p(z, A\cup t)}{p(y, A\cup t)}$. Instead, if $y\sim_p t$, then by Consistency of Revealed Preferences $x\bowtie_p y$ and $y\bowtie_p z$ imply $x\bowtie_p t$ and $z\bowtie_p t$. Then by Generalized Independence of Symmetric Alternatives, $\frac{p(x, A)}{p(z, A)}=\frac{p(x, A\cup t)}{p(z, A\cup t)}$.\smallskip

\noindent\textbf{Case 2.} $y\not\in A$. By Generalized Independence of Symmetric Alternatives, $x\bowtie_p y$ and $z\bowtie_p y$ imply $\frac{p(x, A)}{p(z, A)}=\frac{p(x, A\cup y)}{p(z, A\cup y)}$ and $\frac{p(x, A\cup t)}{p(z, A\cup t)}=\frac{p(x, A\cup y\cup t)}{p(z, A\cup y\cup t)}$. Now by Case 1, $\frac{p(x, A)}{p(z, A)}=\frac{p(x, A\cup y)}{p(z, A\cup y)}=\frac{p(x, A\cup y\cup t)}{p(z, A\cup y\cup t)}=\frac{p(x, A\cup t)}{p(z, A\cup t)}$.

\noindent\textbf{Step 3.} Let $X/\simeq_p\equiv\{X_i\}^n_{i=1}$. Since $\simeq_p$ is reflexive, transitive, and symmetric, $\{X_i\}^n_{i=1}$ is a partition of $X$ such that for any $x_i, x'_i\in X_i$ and $x_j\in X_j$, $x_i\simeq_p x'_i$ and $x_i\not\simeq_p x_j$. Moreover, by the definition of $\simeq_p$, for any $i\le K$, there is some $j\le n$ such that $E_i\subseteq X_j$. Hence, without loss of generality, let $X_i=\bigcup^{t_i}_{s=1} X^s_i$ such that for any $s\le t_i$, $X^s_i=E_l$ for some $l\le K$.\medskip

\noindent\textbf{Step 4.} For any $x, y\in X_i$, $x \bowtie_p y$ if and only if $x\not\sim y$. Hence, the first two parts of Generalized Independence of Symmetric Alternatives are equivalent to  Independence of Symmetric Alternatives when $p$ is restricted on $X_i$. Hence, by Theorem 1, $p$ is an NSC on $X_i$ with some $(u_i, v_i, \{X^s_i\}^{t_i}_{s=1})$. Since $X_i$ and $X_j$ are disjoint for each $i, j$ with $j\neq i$, without loss of generality, we can say that $p$ is an NSC on $X_i$ with the same $(v, u)$. 

\noindent\textbf{Step 5.} Since there are $a, b, c\in X$ such that $a\not\simeq_p b$, $b\not\simeq_p c$, and $a\not\simeq_p c$, we have $n\ge 3$. Take any $A\in\mathscr{A}$. Let $A_i=A\cap X_i$ for each $i\le n$. Take any $a\in A_i, b\in A_j$, and $x\in A_k$. Note that $a\not\simeq_p x$ and $b\not\simeq_p x$. Then by Generalized Independence of Symmetric Alternatives, we have 
\[\frac{p(a, A)}{p(b, A)}=\frac{p(a, A\setminus\{x\})}{p(b, A\setminus\{x\})}.\]
Then by Steps 5-8 of Theorem 1 (also recall Equation (12)), there is a function $w:2^X\to \mathds{R}_{+}$ such that $\frac{p(A_i, A)}{p(A_j, A)}=\frac{w(A_i)}{w(A_j)}$. In other words, $p(A_i, A)=\frac{w(A_i)}{\sum_{j} w(A_j)}$. Since $p$ is an NSC on $X_i$, we also have \[p(a, A_i)=\frac{v(A_i\cap X^s_i)}{\sum_{l} v(A_i\cap X^l_i)}\,\frac{u(a)}{\sum_{b\in A_i\cap X^s_i} u(b)}\] when $a\in A_i\cap X^s_i$. Finally,
\[p(a, A)=p(A_i, A)\,p(a, A_i)=\frac{w(A_i)}{\sum_{j} w(A_j)}\,\frac{v(A_i\cap X^s_i)}{\sum_{l} v(A_i\cap X^l_i)}\,\frac{u(a)}{\sum_{b\in A_i\cap X^s_i} u(b)}.\]

\end{document}